\documentclass[11pt, letterpaper]{article}

\usepackage{amssymb,fullpage}
\usepackage[cmex10]{amsmath}
\usepackage{graphicx}
\usepackage{verbatim}
\usepackage{url,xspace,times}
\usepackage{cite}
\usepackage{caption}
\usepackage{subcaption}
\usepackage{multirow}
\usepackage{multicol}
\usepackage{theorem,latexsym}
\usepackage{amsmath,amssymb,enumerate}
\usepackage{algorithm}

\usepackage{algorithmicx,algpseudocode}

\usepackage{subcaption}

\allowdisplaybreaks

\newcommand{\sse}{\subseteq}

\newcommand{\opt}{\ensuremath{\mathsf{Opt}}\xspace}

\def\F{{\cal F}}
\newcommand{\cC}{\mathcal{C}}

\newcommand{\cP}{\mathcal{P}}

\def\I{{\cal I}}

\newcommand{\E}{\mathbb{E}}

\newcommand{\ignore}[1]{}

\newcommand{\costsingle}{\ensuremath{\beta}\xspace}
\newcommand{\congsingle}{\ensuremath{\gamma}\xspace}

\newcommand{\load}{\ensuremath{\mathsf{load}}\xspace}
\newcommand{\asg}{\ensuremath{\mathsf{terms}}\xspace}

\newcommand{\eevrp}{\textsf{NEERP}\xspace}

\newcommand{\ncnd}{\textsf{MCNC}\xspace}
\newcommand{\scnd}{\textsf{SSNC}\xspace}

\newtheorem{lemma}{Lemma}
\newtheorem{theorem}[lemma]{Theorem}

\newtheorem{definition}[lemma]{Definition}
\newtheorem{corollary}[lemma]{Corollary}

\newtheorem{claim}[lemma]{Claim}

\newenvironment{proof}{\vspace{-0.05in}\noindent{\bf Proof:}}%
        {\hspace*{\fill}$\Box$\par}

\begin{document}

\title{Cluster Before You Hallucinate:  Node-Capacitated Network Design and Energy Efficient Routing\thanks{A preliminary version of this paper appeared in the proceedings of ACM Symposium on Theory of Computing (STOC) 2014.}}


\author{
 Ravishankar Krishnaswamy\thanks{Microsoft Research Vigyan Building, 9, Lavelle Road, Bangalore 560018. Email: ravishankar.k@gmail.com.}  \and
 Viswanath Nagarajan\thanks{Department of Industrial and Operations Engineering, University of Michigan, Ann Arbor, MI 48109, USA. Email: viswa@umich.edu. Supported in part by NSF grants CCF-2006778 and CMMI-1940766. } \and
 Kirk Pruhs\thanks{ Computer Science Department, University of Pittsburgh, Pittsburgh, 15260, PA, USA. Email: krp2@pitt.edu. Supported in part by NSF grants CCF-1907673, CCF-2036077, CCF-2209654
and an IBM Faculty Award. }  \and
 Clifford Stein\thanks{ Department of Industrial Engineering and Operations Research, Columbia University, New York, 10027, NY, USA. Email: cliff@ieor.columbia.edu. Supported in part  by NSF grant  CCF-2218677, ONR grant ONR-13533312, and by the Wai T. Chang Chair in Industrial Engineering and Operations Research at Columbia University.}
}

\maketitle

\begin{abstract}
We consider the following  node-capacitated network design problem. The input is an undirected graph, set of demands, uniform node capacity and arbitrary node costs. The goal
is to find a minimum node-cost subgraph that supports all demands  concurrently subject to the node capacities. We consider both single and multi-commodity demands, and provide the first poly-logarithmic approximation guarantees. 
For single-commodity demands (i.e., all request pairs have the same sink node), we obtain  an $O(\log^2 n)$ approximation to the cost with an $O(\log^3 n)$ factor violation in node capacities.
 For multi-commodity demands, we obtain  an $O(\log^4 n)$ approximation to the cost with an $O(\log^{10} n)$ factor violation in node capacities. We use a variety of techniques, including single-sink confluent flows, low-load set cover, random sampling and cut-sparsification. We also develop new techniques for clustering multicommodity demands into (nearly) node-disjoint clusters, which may be of independent interest.

Moreover, this network design problem has applications to energy-efficient 
virtual circuit routing. In this setting, there is a network
of routers that are speed scalable, and that may be shutdown
when idle. We assume  the standard model for power: the power consumed by a router with load (speed) $s$ is $\sigma + s^\alpha$ where $\sigma$ is the static power and the exponent $\alpha > 1$. We obtain the first poly-logarithmic approximation algorithms for this problem when speed-scaling occurs on nodes of a network. 

\end{abstract}

\section{Introduction}
Network design problems involve finding a minimum-cost subgraph of a given graph while satisfying certain demand requirements. Classic examples include Steiner  tree, Steiner forest, survivable network design and buy-at-bulk network design. Good approximation algorithms are known for all these basic network design problems~\cite{ByrkaGRS13,AgrawalKR95,GoemansW95,Jain98,GMM01,CHKS06}. However, these problems become significantly harder in the presence of capacities, and  much less is known for {\em capacitated} network design problems. 
 In this paper, we study a natural node-capacitated network  design problem and provide the first poly-logarithmic approximation algorithms for it.

In the {\em multicommodity node-capacitated network design problem} (\ncnd), there is  an
undirected graph $G = (V,E)$  where each node $v\in V$  has cost $c_v\ge 0$ and uniform capacity $q$.  
There are also $k$ request-pairs of the form $(s_i, t_i, d_i)$ where $s_i\in V$ is the source, $t_i\in V$ is the sink and $1\le d_i\le q$ is the demand.  A feasible solution is a
subset of nodes $U \sse V$ such that the graph $G[U]$ induced on
nodes $U$ (where each  node has capacity $q$) can {\em concurrently} support $d_i$ units of unsplittable flow between
$s_i$ and $t_i$ for each request-pair
$i\in[k]$.  
The objective is to minimize the total cost $c(U) :=
\sum_{v \in U} c_v$ of the solution.  Instead of requiring unsplittable flows,  one could alternatively ask for a splittable (i.e., fractional) flow  for the demands. However,   this does not change the problem significantly. In fact, our approximation guarantees also hold relative to an optimal solution for splittable  flows.

Our algorithms will find
bicriteria approximations, where the solution is allowed to violate
the capacity constraints by some factor.  A  $(\beta, \gamma)$ bicriteria approximation algorithm for \ncnd finds a solution $U$ such that (i) the cost $c(U)$ is at most $\beta$ times the optimum and (ii)  all request-pairs can be routed concurrently in $G[U]$ using capacity at most $\gamma\cdot q$ at each node.

Other than being a natural theoretical model, \ncnd has applications in energy-efficient routing. Indeed, this was our primary motivation to study \ncnd. Improving the energy efficiency of telecommunication (telecom) networks is an important practical issue. 
In their 2020 report~\cite{McKinsey} ``The case for committing to greener telecom networks" McKinsey reported that
telecom operators account for 2 to 3 percent of total global energy demand, often making them some of the most energy-intensive companies in their geographic markets. But the report noted that
all operators have considerable scope to cut energy costs and consumption, 
with many operators having the potential to reduce energy consumption by at least 15 to 20 percent.
Further improved optimization policies was listed as one of the four key energy reduction opportunities. In this paper, we consider virtual circuit routing, which is used by several network protocols to achieve
reliable communication~\cite{KR09}.

Formally, we consider virtual circuit routing protocols (where each
connection is assigned a fixed route in the network) with an objective
of minimizing energy, in a network of routers that (i) are speed
scalable, and (ii) may be shutdown when idle.  We use the standard
model for a router's power-rate curve, which is  the same as  
in~\cite{AndrewsFZZ10a,AndrewsAZ10,Medalg,AntoniadisIKMNPS14}. In this model, the energy cost incurred by a router operating at speed $x$ (which is assumed to be the total traffic passing through the router) is given by 
\begin{equation}\label{eq:energy}
f(x)=\left\{ \begin{array}{ll}
0&\mbox{ if }x=0\\
\sigma+x^\alpha& \mbox{ if }x>0
\end{array}\right..
\end{equation} 
Above, parameter $\sigma$ is the static power (that is always used when the router is turned on) and parameter $\alpha>1$ specifies the energy inefficiency of the router. The value of $\alpha$ is in the range $[1.1, 3]$ for
essentially all technologies~\cite{B-etal00,WAT09}.  We assume that all network components are homogeneous: so $\sigma$ and $\alpha$ are uniform across all routers. This is also the setting in several prior works~\cite{AndrewsAZ10,Medalg,AntoniadisIKMNPS14}.  

In this paper, we obtain the first poly-logarithmic approximation algorithms for virtual circuit routing with speed-scalable components at {\em nodes} of the network. All previous papers considered the simpler setting where speed-scaling occurred at edges. Although speed scalable edges (corresponding to network links) are plausible,
it is more realistic that speed-scaling occurs  at nodes (corresponding to routers). 

Formally, in the 
{\em node-cost energy efficient routing problem} (\eevrp), we are given  an undirected graph $G = (V,E)$, with non-negative multipliers on nodes $\{c_v \}_{v\in V}$ and a uniform energy cost function \eqref{eq:energy}.  
We are also given a collection of $k$ request-pairs of the form $(s_i, t_i, d_i)$ where for each $i \in [k]$, $s_i\in V$ is the source node, $t_i \in V$ is the sink node and $d_i\ge 1$ is the demand.   The goal in \eevrp is to find a path
$P_i$ connecting $s_i$ and $t_i$ for each $i\in [k]$ so as to minimize the overall energy cost:
$$\sum_{v\in V} c_v\cdot f\left(\sum_{i:v\in P_i} d_i\right).$$

It turns out that \eevrp reduces 
to the capacitated network design problem \ncnd, as stated in the following result. This reduction is implicit in \cite{AndrewsAZ10} where it was applied to the edge-version, and  we provide a proof in  Appendix~\ref{app:reduction} for completeness.  
\begin{theorem}[\cite{AndrewsAZ10}] \label{thm:reduction}
If there is a $(\beta,\gamma)$ bicriteria approximation algorithm  for \ncnd then there is an $O(\beta\cdot \gamma^\alpha)$-approximation algorithm for \eevrp. 
\end{theorem}

 \medskip

\noindent {\em Preliminary simplifications.} 
We refer to the set $\{s_i\}_{i=1}^k\cup \{t_i\}_{i=1}^k$ of all sources and sinks in \ncnd as {\em terminals}. 
 We assume (without loss of generality) that (i) all terminals are distinct, i.e., each node in $V$ is the source or sink of at most one request and (ii) each terminal is a leaf node, i.e.,  has degree one. This can be ensured by adding $2k$ new terminals of cost zero, where each new terminal node ($s_i$ or $t_i$) is connected only to the original terminal. So  the number of nodes $n\ge 2k$.  
In some applications, we may also have $n\gg k$ (which is common in   network design problems). So, we state our approximation ratios in terms of both $n$ and $k$. 

We also note that, without loss of generality,  {\em zero-cost}  nodes in \ncnd may have capacity that is any integral multiple of $q$. 
To see this, consider any zero-cost node $v\in V$ (i.e., with $c_v=0$) having capacity $z\cdot q$ where $z\ge 1$ is an integer.  Then, we simply introduce $z$ copies $v_1,\cdots v_z$ of node $v$, each having uniform capacity $q$ and zero cost. As all the copies have zero cost, it is clear that the two \ncnd instances are equivalent. 
We note that this reduction increases the number of nodes, but we always have $z\le k$ as the total demand in any instance is at most $kq$; so the number of nodes in the new instance is at most $nk$.

\medskip

\noindent {\em Single-sink node-capacitated network design (\scnd)}. 
We also consider separately the {\em single-sink} special case of \ncnd, where there is a common node $t\in V$ with $t_i=t$ for all $i\in [k]$. The sink node $t$ is assumed to have zero cost; this is without loss of generality as $t$  must be included in any feasible solution.  Moreover, the capacity of $t$ is $kq$ so that all demands can be routed into it. 
We also assume  that each source is a distinct leaf node.
The single-sink problem serves as a simpler setting to explain our techniques, and is also used  in the multicommodity algorithm.


\subsection{Our Results and Techniques}

Our first main result is:
 \begin{theorem}
\label{thm:single-sink}
There is an $(O(\log^{2} n), O(\log^3 n))$ bicriteria approximation algorithm for single-sink node-capacitated network design.  
\end{theorem}
In order to motivate our approach, we illustrate two corner-cases which are interesting in their own right. If the total demand is smaller than the capacity,  i.e., $\sum_{i} d_i < q$, then the problem reduces to computing a minimum-cost node-weighted Steiner tree, for which there are $O(\log k)$-approximation algorithms~\cite{KleinR94,Guha99improvedmethods}. At the other end of the spectrum, if each demand $d_i$ is large, i.e., $\min_i d_i =\Omega(q)$, then each of these requests essentially has to route its demand on a disjoint path, and this problem can be solved by using techniques from low-congestion routing \cite{RaghavanT87}. (See also Appendix~\ref{app:small-q}.) 
 Prior results for the {\em edge-capacitated} problem~\cite{Medalg,AndrewsAZ10,AntoniadisIKMNPS14} were based on a combination of these ideas, and can be summarized as follows: (i) choose an
 approximately min-cost Steiner tree  $T$ connecting all the
sources and the sink, (ii) partition  $T$ into \emph{edge-disjoint} subtrees (which we call clusters) having total demand $\approx q$ in each, (iii) choose one ``leader'' in each cluster and aggregate all demand in the cluster at the leader and (iv) route $q$ units of flow from each leader to the sink $t$ using disjoint paths. 
 The overall edge-congestion is bounded because the clusters are edge-disjoint  and  the (disjoint)  path chosen by each leader suffices to route the entire demand in that cluster (which is at most $q$). 
 A crucial ingredient in this approach is that  any tree  can be partitioned into \emph{edge-disjoint} subtrees/clusters containing $\approx q$ demand each.

However, in  the node-capacitated setting, there may not exist a {\em node-disjoint} clustering of the  minimum Steiner tree into subtrees of $\approx q$ demand each! For example, the tree $T$ could just be a star with all the sources and sink as leaves, which means that the center node will appear in every cluster. So, instead of partitioning a min-cost Steiner tree into clusters (which may not be possible), we  directly aim to find low-cost node-disjoint clusters.  However, it is not {\em a priori} clear that such a clustering must always exist. Our first step in  Theorem~\ref{thm:single-sink} is to prove the {\em existence} of node-disjoint clusters of cost at most the optimal \scnd value where each cluster has $O(\log n)\cdot q$ demand. 
 This proof relies on the existence of single-sink {\em confluent flows}~\cite{ChenKLRSV07}. In fact, we show that each such cluster can be rooted at a neighbor of the sink so that routing from each cluster to the sink is trivial.  
 Our second step in Theorem~\ref{thm:single-sink} is to {\em efficiently find} such a clustering.  
 We achieve this by formulating the single-sink clustering problem as an instance of \emph{low load set cover} \cite{BabenkoGGN12}. Here, each subtree with $O(\log n)\cdot q$  demand is a ``set''  and we need to pick a min-cost collection of sets such that the  number of sets containing any node is bounded (which will ensure approximate node-disjointness). The approximation algorithm for low load set cover from~\cite{BabenkoGGN12} requires a subroutine for the
related ``minimum ratio'' problem, for which we obtain an $O(\log n)$-approximation algorithm using the \emph{partial node weighted Steiner tree} problem~\cite{KonemannSS12,MossR07}. These details are  presented in \S~ref{sec:ss}.

Our second main result is:
\begin{theorem}
\label{thm:multicomm}
There is an $(O(\log^{2} n \log^2 k), O(\log^{6} n \log^4k))$ bicriteria approximation algorithm for  multicommodity  node-capacitated network design.  
\end{theorem}
We note that an $\Omega(\frac{\log\log n}{\log\log\log n})$ factor violation in the node-capacity is necessary for any non-trivial approximation on the cost, due to the hardness of the undirected congestion minimization problem~\cite{AndrewsZ07}.

Our high-level approach is similar to that for the single-sink case. First, we find a {\em clustering} of all source and sink nodes into nearly node-disjoint subtrees of small cost such that each cluster has at most $q\cdot polylog(n)$ demand inside. Next, we find a {\em routing} of demands across different clusters (from sources to sinks) while incurring low node-congestion. However, both these steps are significantly more complicated than the single-sink case, as outlined next. 
 
For multicommodity clustering,  as in the single-sink case, we need to prove both the existence and computation of  (nearly) node-disjoint clusters.  However, there is no multicommodity notion of confluent flow, which was used crucially in the single-sink existence proof. Moreover,  the low-load set-cover approach is not applicable either because the ``minimum ratio'' problem in the multicommodity case is at least as hard to approximate as the {\em dense-$k$-subgraph} problem, which is believed to not admit any poly-logarithmic approximation~\cite{FPK01,BCCFV10,Manurangsi17}.  We also need to modify the notion of an allowed cluster in the multicommodity case. Ideally, we would like each cluster to be ``heavy'', i.e., having demand at least $q$ (and at most $q\cdot polylog(n)$), which  is useful in the subsequent routing step. However, this  may not always be possible: so we also allow ``internal'' clusters where a constant fraction of the demand in the cluster comes from requests with {\em both} source and sink in that cluster. Then, we obtain a low cost clustering where each cluster is either heavy or internal. We also ensure that the clusters have low node congestion, i.e., each node appears in at most  $polylog(n)$ many clusters.  
Our algorithm constructs these clusters in an iterative manner, where we use the  single-sink algorithm (Theorem~\ref{thm:single-sink}) in each iteration. We start off with  each terminal being a singleton cluster, and  continue merging clusters until each cluster is either heavy or internal. Crucially, we prove that the \scnd instances solved in each iteration have low cost by producing a ``witness solution'' using the optimal \ncnd solution. We then use the \scnd solutions to merge clusters so that the number of clusters reduces by a constant factor in each iteration:  this implies that  $O(\log k)$ iterations suffice. The actual algorithm is more subtle and we only end up clustering a constant fraction of the total demand. See \S\ref{subsec:mc-cluster} for a more detailed overview of the clustering algorithm.

For multicommodity routing, we consider two cases depending on whether there are more demands in internal or heavy clusters. If a constant fraction of the demand is contained in internal clusters then we do not have to route across clusters: we just route all ``internal'' demands using the respective subtrees. The harder case is when a 
constant fraction of the demand is in heavy clusters. Here, we  find a low-cost routing across heavy clusters using a sampling/hallucination based approach from the edge-capacitated problem~\cite{AntoniadisIKMNPS14}. However, unlike the edge version \cite{AntoniadisIKMNPS14}, in the node version we need to drop some demands in the routing step. This is required to ensure that the min-cut in  the demand graph is large, which in turn is needed for the cut-sparisification result \cite{karger1999random} that we use. See \S\ref{subsec:mc-routing} for a more detailed overview of the routing algorithm.

Finally, after combining the clustering and routing steps, we obtain a solution that can support a constant fraction of the total demands. So, we need to apply these steps recursively on the remaining demands to complete the proof of Theorem~\ref{thm:multicomm}.

We also note that the approximation ratios in Theorems~\ref{thm:single-sink} and \ref{thm:multicomm} can be strengthened to be relative to an optimal {\em splittable} routing: see Appendix~\ref{app:splittable}. 

 Using Theorems~\ref{thm:single-sink} and \ref{thm:multicomm} along with the reduction in Theorem~\ref{thm:reduction}, we obtain:
 \begin{corollary}
\label{thm:intro-thm-multicast}
There is an $O(\log^{3 \alpha +2} n)$-approximation algorithm for
the  single sink node-cost energy-efficient routing problem.
\end{corollary}

\begin{corollary}
\label{thm:intro-thm-unicast}
There is an $O(\log^{10 \alpha +4} n)$-approximation algorithm for the  multicommodity node-cost energy-efficient  routing problem. 
\end{corollary}

\subsection{Related Work}

Approximation algorithms for the edge-capacitated version of \ncnd have been studied previously in ~\cite{AndrewsFZZ10a,AndrewsAZ10,Medalg,AntoniadisIKMNPS14}. The node-capacitated problem that we study is more general, and we obtain the first approximation algorithms.    
A key challenge that needs to be addressed  in these results is that the problem has similarities to both convex and concave cost flows. When the capacity $q$ is small, the \ncnd problem is similar to convex-cost flow, where it is preferable  to spread flow over disjoint paths. On the other hand, when the capacity $q$ is large, \ncnd is similar to concave-cost flow, where one prefers to aggregate flow. In \cite{AndrewsAZ10}, the authors showed that these competing forces (to spread-out or aggregate flow) can be
``poly-log-balanced'' by giving a bicriteria  poly-logarithmic
approximation algorithm for the multicommodity edge version of the problem. Moreover, 
\cite{AndrewsFZZ10a} showed  an $\Omega(\log^{1/4} n)$ inapproximability result for the edge
version, under standard complexity theoretic assumptions. 
Later, \cite{AntoniadisIKMNPS14} obtained an improved $(O(\log n), O(\log n))$ bicriteria approximation algorithm for edge-capacitated  \ncnd.  In fact, \cite{AntoniadisIKMNPS14}  also studied  the online version (where requests arrive over time) and obtained an $(O(\log n), O(\log^2 n))$ bicriteria competitive ratio. A key technique in \cite{AntoniadisIKMNPS14} was a random-sampling idea, where  each request $i$ ``hallucinates'' that it wants to route $q$ units with probability $\approx \frac{d_i}q$.  We also make use of this idea in our paper.

The \eevrp problem has also been studied
in the special case that speed scaling occurs on the edges instead of
the nodes. 
 As noted earlier, it is more realistic to have speed-scalable nodes   rather than edges.  
Presumably, the assumption in these previous papers that speed scaling occurs on the edges was motivated
by reasons of mathematical tractability, as
network design problems with edge costs are
usually easier to solve
than the corresponding problems with
node costs. 
The paper \cite{AndrewsAZ10} obtained a $\log^{O(\alpha)}n$-approximation algorithm for the edge-cost  \eevrp.  The paper \cite{Medalg} considered the single-sink special case (with edge costs) and obtained an
$O(1)$-approximation algorithm and $O(\log^{2 \alpha+1} n)$-competitive randomized
online algorithm. Later,
\cite{AntoniadisIKMNPS14} obtained a simple  $O(\log^\alpha n)$-approximation algorithm for
the multicommodity edge version, which was also extended to an
$\tilde O(\log^{3\alpha+1} n)$-competitive randomized online algorithm.

We note that the \ncnd problem is  a special case of a very general model, called {\em fixed-charge network design} that has been studied extensively in the operations research literature, see e.g. \cite{KimP99,Costa05,HewittNS13}. The focus in these papers has been on solving the problem exactly, which is different from our goal of polynomial-time approximation algorithms. 

An $O(\log k)$-approximation algorithm for the basic  {\em node-weighted} Steiner tree problem was obtained in ~\cite{KleinR94}, which is also the best possible approximation ratio (as set cover is a special case). This contrasts with the usual (edge-weighted) Steiner tree, for which constant-factor approximations are known~\cite{ByrkaGRS13}.  Our algorithm also makes use  of the \emph{partial} node weighted Steiner tree  (PNWST) problem, where we only want to connect a certain number of terminals. An $O(\log n)$ approximation algorithm for PNWST was obtained in \cite{KonemannSS12,MossR07}.

Buy-at-bulk  network design is also somewhat related to our model. Here, the cost on a network element (edge or node) is a {\em concave} function of the load through it. Poly-logarithmic approximation algorithms are known for both  edge-weighted~\cite{AA97,GMM01,CHKS06} and node-weighted cases~\cite{CHKS07,ACSZ07}. The paper \cite{Andrews04} also showed  
poly-logarithmic hardness of approximation for  
buy-at-bulk network design.
From a technical standpoint, the hallucination idea used in~\cite{AntoniadisIKMNPS14} and also in our algorithm,
is  similar to the Sample-Augment framework in~\cite{GuptaKPR07} for solving Buy-at-Bulk problems.
However, our algorithm analysis is quite different from those
for Buy-at-Bulk, and is more similar in spirit to the analysis of cut-sparsification algorithms~\cite{karger1999random,SpielmanT11,fung2011general}.

The \emph{survivable network design} problem (SNDP) is a different  (but well-studied) multicommodity network design problem. Here,
the goal is to select a minimum-cost subgraph that can route a set of
demands {\em individually}, i.e., each demand should be routable in the subgraph (just by itself).  
A $2$-approximation algorithm is known for SNDP with edge-connectivity requirements~\cite{Jain98}. 
The node-connectivity SNDP has also been studied extensively, with the best approximation ratio being  $O(k^3 \log n)$ for edge costs~\cite{CK09} and $O(k^4 \log^2 n)$ for node costs~\cite{Nutov12}; here $k$ is the largest demand.
Vertex-connectivity SNDP  is also $\Omega(k^\epsilon)$-hard to approximate~\cite{CCK08}. 
There has also been some  work on capacitated SNDP~\cite{CFLP00, CCKK11, CKLN, HKKN}. We note that capacitated SNDP differs significantly from \ncnd  because the goal in
SNDP is to route  each request-pair in isolation, whereas our goal  is to route all requests concurrently.

Very recently (after the conference version of this paper), \cite{EmekKLS20,NagarajanW21} considered the \eevrp problem with {\em non-uniform} cost functions, where the $\sigma$ and $\alpha$ parameters in \eqref{eq:energy} are different across nodes.  In fact, their results apply to a larger class of ``generalized network design'' problems, which includes multicommodity routing on directed graphs. 
 The paper \cite{EmekKLS20} gave an $O\left(\max_{v\in V} \sigma_v^{1/\alpha_v}\right)$-approximation algorithm for non-uniform \eevrp, where $\sigma_v$ and $\alpha_v$ are the cost parameters for each node (and $\alpha_v$s are constant). Moreover, \cite{NagarajanW21} obtained an online algorithm for non-uniform \eevrp with the same competitive ratio. We note that these results are incomparable to Corollary~\ref{thm:intro-thm-unicast}: we obtain approximation ratios that are poly-logarithmic in the input size ($n$ and $k$), whereas these results in \cite{EmekKLS20,NagarajanW21} have a polynomial dependence (albeit on the cost parameters). 


\newcommand{\cfinal}{K}
\renewcommand{\d}{\tilde{d}}

\section{Single-Sink Node-Capacitated Network Design} 
 \label{sec:ss}
\def\snd{{\sf SSNC}\xspace}
 The input to the  \snd problem  consists of an undirected graph $G = (V,E)$,
with $|V| = n$, and a collection of $k$ sources ${\cal D} = \{s_i \; |
i \in [k]\}$ with respective demands $\{1\le d_i \leq q \; | i \in
[k]\}$. Recall that 
 each source node has degree one. There is a specified sink $t\in V$ to which each source $s_i$
wants to send $d_i$ units of flow unsplittably. Each node $v \in
V\setminus \{t\}$ has a cost $c_v$ and uniform capacity $q$; the sink
$t$  has  zero cost and capacity $kq$ (so  all demands can be routed into it). Recall that  zero-cost nodes in \ncnd (and \scnd) are allowed to have larger capacity than $q$. 
 The
output is a subset of nodes $V' \sse V$ such that the graph $G[V']$
induced by the nodes $V'$ can concurrently support an unsplittable
flow of $d_i$ units from each source $s_i$ to the sink $t$.  The
objective is to minimize the total cost $c(V') = \sum_{v \in V'} c(v)$. We will also refer to the nodes $\{s_i\; | i
\in [k]\}$ as {\em terminals}. In our analysis, we use \opt to denote the cost of the optimal \scnd solution.

A simple but important notion is that of a single-sink cluster, defined below. 
\begin{definition}[\scnd Cluster] \label{defn:s-cluster}
A {\em cluster} is any subtree  of graph $G$  containing the sink $t$. 
The demand of the cluster is the total demand of all sources in it. 
\end{definition}

The key step in our single-sink algorithm is to find a collection of nearly node-disjoint clusters, each assigned 
roughly $q$ demand. 
An important step is to even show the existence
of such clusters, which we do in
\S\ref{subsec:ss-exists}. The existence argument is based on  {\em single-sink confluent flows}~\cite{ChenKLRSV07}.  We then give an algorithm for
finding such clusters in \S\ref{subsec:ss-find}. This algorithm relies
(in a black-box fashion) on two other results: an $O(\log n)$-approximation algorithm for {\em partial node-weighted Steiner
  tree}~\cite{KonemannSS12,MossR07}, and a logarithmic bicriteria
approximation for {\em low load set cover}~\cite{BabenkoGGN12}. At a
high level, we model a set cover instance on the graph, where any cluster is a set, and the goal is to find a minimum cost set cover of all
terminals such that no node is in too many sets. The algorithm
of~\cite{BabenkoGGN12} requires a {\em min-ratio} oracle, for which
we use the partial node-weighted Steiner tree algorithm. Finally, we just select all the nodes in the clusters as our solution. The node congestion can be bounded using the fact  that each cluster has roughly $q$ demand and that the clusters are nearly disjoint.

\smallskip
\noindent {\bf Confluent Flow. } 
 Consider any $n$-node directed graph with sink node $t$, sources $\{s_i\}_{i=1}^k$ with demands $\{d_i\}_{i=1}^k$ and uniform node capacity $q$ at all nodes except the sink (which has infinite capacity). Again, we assume that each demand is at most $q$. A flow is said to be {\em confluent} if for every node $u$ there is at most one edge $(u,v)$ out of $u$ that carries positive flow. Note that the edges carrying positive flow in any confluent flow correspond to an arborescence directed toward the sink $t$.  
\begin{theorem}[Theorem~20 in \cite{ChenKLRSV07}]\label{thm:conflu}
Consider any directed graph as above with a splittable routing $\F^*$ that sends $d_i$ units from each source $s_i$ (for $i\in [k]$) to sink $t$, while respecting node capacities.  Then, there is a  confluent flow $\F$ that routes all demands where the total flow through any node (other than $t$) is at most $(1+\ln n)q$.
\end{theorem}
 The multiple sinks referred to in \cite{ChenKLRSV07} correspond to the in-neighbors of our single sink $t$ (which are at most $n$ in number).

\subsection{Existence of Good Clustering} \label{subsec:ss-exists}
We first show that there exists a ``good'' clustering of the source nodes into node-disjoint clusters.
\begin{lemma}
\label{lem:cluster}
Given any instance of \scnd with optimal cost \opt, there exists a collection $\{ T_i  \}_{i=1}^g$ of clusters such that
\begin{enumerate}
\item[(i)] The demand of each cluster $T_i$ 
is at most $(1+\ln n)\cdot q$.
\item[(ii)] Every source lies in some cluster.
\item[(iii)] The clusters are node-disjoint except at $t$.  
\item[(iv)] The total cost is $\sum_{i=1}^g \sum_{v\in T_i} c_v \le \opt$. 
\end{enumerate}
\end{lemma}

\begin{proof}
Let $V^*\sse V$ denote the set of nodes in an optimal solution and  $\F^*$ denote an optimal flow for the sources ${\cal D}$. Note that $\F^*$ sends at most $q$ flow through each node (except $t$). We now apply Theorem~\ref{thm:conflu} on the  graph induced on $V^*$, to obtain a {\em confluent flow} $\F$ where the flow through each node (other than $t$) is at most $q(1+\ln n)$. Recall that $\F$ corresponds to an arborescence ${\cal T}$ directed toward the sink $t$. Let $\{r_i\}_{i=1}^g$ denote all neighbors of $t$ contained in arborescence ${\cal T}$. For each $i\in [g]$, let $T_i$ denote the subtree of ${\cal T}$ rooted at  $r_i$, along with the edge $(t,r_i)$. Note that $\{T_i\}_{i=1}^g$ are node-disjoint except at $t$. Moreover,  the total demand in each subtree $T_i$ is at most $q(1+\ln n)$ because all of these demands pass through node $r_i$. Finally, ${\cal T}=\cup_{i=1}^g T_i$ contains all the sources as the confluent flow $\F$ routes every demand.  

We claim that the clusters $\{T_i\}_{i=1}^g$ satisfy all the conditions in the lemma.
Conditions (i)-(iii) follow directly from the above construction.  
Condition (iv) follows from (iii) and the fact that all nodes of $T_i$ are contained in $V^*$.
\end{proof}

\subsection{Finding a Good Clustering}\label{subsec:ss-find}

The previous subsection only establishes the existence of a good
clustering; in this subsection we explain how to efficiently find such a
clustering.  

\begin{lemma}
	\label{lem:find-cluster}
There is an efficient algorithm that, for any instance of \scnd with optimal cost  \opt,  
  finds a collection of clusters $\{ T_i  \}_{i=1}^g$ such that:
	\begin{enumerate}
\item[(i)] The demand of each cluster $T_i$ 
is at most $(1+\ln n)\cdot q$.
		\item[(ii)] Every source lies in some cluster.
		\item[(iii)] Every node in $V\setminus \{t\}$ appears in at most $O(\log^2n)$ clusters.  
		\item[(iv)] The total cost is $\sum_{i=1}^g \sum_{v\in T_i} c_v \le O(\log^2n) \cdot \opt$. 
	\end{enumerate}
\end{lemma}

 Given this clustering, our final solution to the \scnd instance is just $\cup_{i=1}^g T_i$. As each cluster contains the sink $t$, there is no need for a separate routing step (from  clusters to $t$). By Lemma~\ref{lem:find-cluster} property (iv), the cost is $O(\log^2n)\cdot \opt$. Moreover, by properties (i), (ii) and (iii), all demands can be routed unsplittably with a total flow of $O(\log^3n)\cdot q$ through any node (other than $t$). This completes the proof Theorem~\ref{thm:single-sink}. It remains to prove~ Lemma~\ref{lem:find-cluster}, which will do now. Our algorithm will use an  approximation algorithm for \emph{low load set cover} (LLSC)~\cite{BabenkoGGN12},  defined next.

\smallskip
\noindent {\bf
Low load set cover (LLSC). }
In this problem,  we are given a set system $(U, {\cal C})$ with elements $U$ and sets ${\cal C}\sse 2^U$,
costs $\{c_v :v\in U\}$, and bound $p\ge 1$.
The cost of any set $S\in {\cal C}$ is $c(S) := \sum_{v\in S}c_v$, the sum of its element costs. We note that the collection ${\cal C}$ may be exponentially large and specified implicitly.
The cost of any collection ${\cal C'} \subseteq {\cal C}$ is $c({\cal C}') := \sum_{S\in {\cal C}'} c(S)=\sum_{S\in {\cal C}'} \sum_{v\in S}c_v$ the sum of its set costs.
We are also given two special subsets of elements: {\em required} elements  $W\sse U$ that need to be covered, and {\em capacitated} elements $L\sse U$. Our LLSC definition is slightly different from that in \cite{BabenkoGGN12} due to the presence of element costs and having to cover only a subset of elements. However, the algorithm and analysis from \cite{BabenkoGGN12}  extend to our formulation  in a straightforward way: see Appendix~\ref{app:llsc}. The goal is to find a minimum cost set cover ${\cal C'} \subseteq {\cal C}$ for the required elements $W$ (i.e. $\cup_{S\in {\cal C'}} S \supseteq W$) such that each capacitated element $e\in L$ appears in at most $p$ sets of ${\cal C'}$. An approximation algorithm for LLSC is given in~\cite{BabenkoGGN12}, which relies on the following subproblem. The {\em min-ratio oracle} for LLSC takes as input, non-negative element-costs $\{\eta_v:v\in U\}$ and a subset $X \subseteq W$ (of already covered required elements), and outputs a  set $S \in {\cal C}$ that minimizes $\frac{\sum_{v \in S} \eta_v}{|S\cap (W \setminus X)|}$. We use the following result on LLSC.
\begin{theorem}[\cite{BabenkoGGN12}] \label{thm:llsc}
Assuming a $\rho$-approximate min-ratio oracle, there is an algorithm for the LLSC problem, that finds a solution of cost  $O(\rho\log |U|)$ times the optimum and which covers each capacitated element  $O(\rho\log |U|)\, p$ times.
\end{theorem}
In other words, this is a bicriteria approximation algorithm that violates both the cost and capacities by an $O(\rho\log |U|)$ factor.

\medskip
\noindent {\bf The \snd problem as LLSC.} We now prove Lemma~\ref{lem:find-cluster} by casting the desired clustering problem as an instance of LLSC. The elements are the nodes $V$  of the original \snd problem. The costs $\{c_v:v\in V\}$ are the node-costs in \snd. The required elements are all the sources $W=\{s_i\}_{i=1}^k$. For any $v\in W$ its demand $d_v :=d_i$ where $v=s_i$ is the corresponding source node. The capacitated elements are $L:=V\setminus \{t\}$ and the bound $p=1$.  Let $Q:=(1+\ln n)\cdot q$. The collection ${\cal C}$ of sets is defined as follows. There is a set corresponding to each cluster  (Definition~\ref{defn:s-cluster}),  having  demand at most $Q$. To reduce notation, we   use $T$ to denote the subtree representing the cluster as well as the nodes in this cluster.   
 By Lemma~\ref{lem:cluster},  the optimal value of this LLSC instance is at most $\opt$. 
 
 Next, we provide an approximation algorithm for the min-ratio oracle for such LLSC instances. Our min-ratio algorithm relies on another known problem: 
 
\smallskip
\noindent {\bf
	Partial node weighted Steiner tree (PNWST). }
The input is an undirected graph $G=(V,E)$ with node-weights $\{\eta_v:v\in V\}$, sink $t\in V$,  rewards $\{\pi_v : v\in V\}$,  and 
target $\tau$.  Both the node-weights and rewards are non-negative.  The objective is to find a minimum node cost
Steiner tree containing $t$ having total reward at least $\tau$.  We will use the following known result.

\begin{theorem}[\cite{KonemannSS12,MossR07}] \label{thm:pnwst}
	There is an $O(\log |V|)$-approximation algorithm for the partial node weighted Steiner tree problem.
\end{theorem}

\begin{lemma} \label{lem:oracle}
There is an $O(\log n)$-approximate min-ratio oracle for the \snd clustering problem.
\end{lemma}

\begin{proof}
 In the min-ratio oracle of the \snd clustering problem, we are given non-negative node weights $\{\eta_v:v\in V\}$ and subset $X \subseteq W$. The goal is to find:
 $$\min_{T \in {\cal C}} \,\, \frac{\sum_{v \in T } \eta_v}{|T \cap (W\setminus X)|}.$$

We will refer to the nodes $W\setminus X$ as {\em new} nodes. Our min-ratio oracle   involves solving several PNWST instances, as defined below.  

For each $\ell=1,2,\cdots |W\setminus X|$, we define an instance $\I_\ell$ of PNWST as follows. 
\begin{itemize}
\item The node-weights are $\{\eta_v:v\in V\}$. 
\item The rewards are 
$$\pi_v = \left\{ 
\begin{array}{ll}
\frac1\ell - \frac{d_v}{2Q} & \mbox{ if }v\in W\setminus X\\
0 & \mbox{ otherwise} 
\end{array}\right., \qquad \forall v\in V.$$
Note that some node-rewards may be negative.   
Let $V':=\{v\in V: \pi_v\ge 0\}$ denote the nodes with  non-negative reward. Note that all nodes in $V\setminus V'$ are leaf nodes (by our assumption that all sources are leaf nodes).  
\item The input graph $G'$ is the subgraph of $G$ induced on nodes $V'$.
\item The target reward is $\tau=\frac12$.
\end{itemize}

Let $T_\ell$ denote the tree obtained from the $\rho=O(\log n)$ approximation algorithm for the PNWST  instance $\I_\ell$. (If instance $\I_\ell$ is infeasible then $T_\ell=\mathsf{NIL}$ and we skip the following steps.) Let $N\sse W\setminus X$ denote the new  nodes covered by tree $T_\ell$ and let 
$F=\sum_{v\in N} d_v$ be their total demand. So the reward of tree $T_\ell$ is $\frac{|N|}{\ell} - \frac{F}{2Q}\ge \frac12$. In other words,
\begin{equation}\label{eq:density-calc}
|N| \ge \frac{Q+F}{2Q}\cdot \ell
\end{equation}

Note that subtree $T_\ell$ may  not be in the set-collection ${\cal C}$ as  its demand $F$ may be more than $Q$. To fix this issue, we will select a subset $N'\sse N$ of nodes in  $T_\ell$ with demand  at most $Q$, and construct a subtree $T'_\ell\in {\cal C}$ that contains $N'$. 

If $F\le Q$ then $N' = N$ and $T'_\ell=T_\ell$. Note that $T'_\ell$ is in ${\cal C}$ as its demand is at most $Q$.  Moreover, $|N'|=|N|\ge \frac{\ell}2$ by \eqref{eq:density-calc}.  So, the cost-to-coverage ratio of $T'_{\ell}$ is at most $2c(T_\ell)/\ell$.  

If $F>Q$, we do the following.  
Starting with a partition of $N$ into singletons (so each part has demand at most $Q$), repeatedly merge any two parts into one if the resulting part has total demand at most $Q$. At the end of this process, we will have $h\le \frac{2F}Q$ parts each with demand at most $Q$. We set $N'$ to be the largest cardinality part in the final partition. So $|N'|\ge \frac{|N|}{h}\ge \frac{Q}{2F} |N|\ge \frac{\ell}{4}$ by \eqref{eq:density-calc}. Further, let $T'_\ell$ be the subtree obtained from $T_\ell$ by removing nodes $N\setminus N'$. Note that $T'_\ell$ is indeed a tree because the deleted nodes $N\setminus N'\sse W$ are all leaves (by our assumption that all sources are leaf nodes). Crucially, $T'_\ell$ is in ${\cal C}$ because its demand is at most $Q$. Moreover, it covers $|N'|\ge \frac{\ell}{4}$ new nodes. So, the cost-to-coverage ratio of $T'_{\ell}$ is at most $4c(T_\ell)/\ell$.

Finally, the min-ratio oracle returns the subtree having the minimum ratio among $\{T'_\ell : 1\le \ell\le |W\setminus X|\}$. 
 We now show that this is a $4\rho$-approximation algorithm.

Let $T^*\in {\cal C}$ denote the min-ratio cluster and $\ell^*=|T^*\cap (W\setminus X)|$ be the number of new nodes in $T^*$. Then, by definition of the rewards $\{\pi_v\}$ in PNWST instance $\I_{\ell^*}$, we have $\pi(T^*)\ge \frac12$ as the total demand in $T^*$ is at most $Q$. However, $T^*$ may not itself be feasible to $\I_{\ell^*}$ as it may not be a subtree of graph $G'$ (which is restricted to the nodes $V'$). Let $T'$ be the subtree of $T^*$ obtained by restricting to the nodes $V'$ of graph $G'$; note that $T'$ is indeed a tree because all nodes $V\setminus V'$ are leaves. Moreover, the reward of $T'$ is at least that of $T^*$  as  nodes in $V\setminus V'$ have negative reward. 
So, $T'$ is a feasible solution to instance $\I_{\ell^*}$, and the optimal value of $\I_{\ell^*}$ is at most $c(T')\le c(T^*)$. Hence,    
$c(T_{\ell^*})\le \rho\cdot c(T^*)$ as $T_{\ell^*}$ is a $\rho$-approximate solution to $\I_{\ell^*}$. So, the ratio of our algorithm's solution  $T'_{\ell^*}$ is at most $\frac{4\rho}{\ell^*}c(T^*)$. 
Using $\rho=O(\log n)$ from Theorem~\ref{thm:pnwst}, we obtain the lemma.
\end{proof}

Finally, the proof of Lemma~\ref{lem:find-cluster} follows from Theorem~\ref{thm:llsc} along with Lemma~\ref{lem:oracle}.

\subsection{Good Clustering from \scnd Solution}
In our  multicommodity algorithm, we will utilize approximate solutions to   \scnd  instances to come up a good clustering. The desired properties  of this clustering are stated in Theorem~\ref{thm:ss-disj-cluster} below. 
 Informally, this result says that given any solution to an \scnd instance with some set of sources $X$ and sink $t$, we can peel out node-disjoint subtrees such that (a) the total demand in any cluster is bounded  and (b) each cluster either has at least two sources or contains a neighbor of the sink $t$. Our clustering algorithm makes use of the following known result on single-sink unsplittable flow. 
 \begin{theorem}[Theorem~3.5 in \cite{DGG}]\label{thm:unsplit}
Consider any directed graph with single sink $t$, sources $X$ having demands $\{d_s:s\in X\}$ and a splittable flow $\F'$ that sends $d_s$ units from each source $s\in X$ to sink $t$, while respecting node capacities.  Then, there is an unsplittable flow $\F$ that routes all demands where the total flow in $\F$ through any node exceeds its original flow (in $\F'$) by  at most 
$\max_{s\in X} d_s$.
\end{theorem}

\begin{theorem}\label{thm:ss-disj-cluster}
Consider any \scnd instance with source-nodes $X$, maximum demand $d_{max}$ and sink $t$. Let $N\sse V$ denote the neighbors of $t$. Suppose $V'\sse V$ is a solution of cost $B$ such that it supports the demand flow from $X$ to $t$ with maximum node capacity of $C$. Then, we can find in polynomial time, a node-disjoint collection of rooted subtrees $\{(r_j,T_j)\}_{j=1}^g$ such that:
\begin{enumerate}
\item every source node appears in some subtree.
\item each subtree $T_j\sse V'\setminus \{t\}$; so the total cost of these subtrees is at most $B$.
\item the total demand in any subtree is at most $C+d_{max}$.  
\item every subtree $T_j$ with root $r_j\not\in N$ contains at least two sources. 
\item for every subtree $T_j$ with root $r_j\not\in N$ we have  $T_j \cap N =\emptyset$.
\end{enumerate}\end{theorem}
\begin{proof}
Consider the network induced on the nodes $V'$ (from the \scnd solution) where each node has capacity $C$. By feasibility of this \scnd solution, there is a splittable flow $\F'$ that sends $d_s$ units of flow from each source $s\in X$ to sink $t$. As this is a single-sink flow, we can ensure that there are no directed cycles in $\F'$. Moreover, we can assume (without loss of generality) that every neighbor of $t$ (i.e., node in $N$) sends flow {\em only} to $t$: this is because we only have capacities at nodes. 
Applying Theorem~\ref{thm:unsplit} to $\F'$, we obtain a flow $\F$ that sends $d_s$ units unsplittably from each source $s\in X$, where the flow through each node (other than $t$) is at most $C+d_{max}$. Moreover, $\F$ does not have any directed cycles because $\F'$ doesn't. We may also assume (without loss of generality) that in $\F$, each node in $N$ (neighbors of $t$) only carries non-zero flow to $t$. Let $E'\sse E$ denote the arcs used in flow $\F$; note that $(V',E')$ is a directed acyclic graph. We index the nodes $V'$ in topological sort order with the sink $t$ having the smallest index $1$. Let $\F(s)$ be the $s-t$ path used to route demand from source $s\in X$. We construct the desired collection of trees as described in Algorithm~\ref{alg:single-comm-cluster}.
Let $g$ denote the number of trees produced. We now show that the rooted subtrees $\{(T_j,r_j)\}_{j=1}^g$ satisfy the claimed properties. 
 
\begin{algorithm}[h!]
\caption {Computing \scnd Clusters \label{alg:single-comm-cluster}}
\begin{algorithmic}[1]
\State initialize sources $Y\gets X$ and $j\gets 1$. 
\While{$Y\ne \emptyset$}
\State\label{step:single-clust-root} let root $r_j\in V'\setminus \{t\}$ denote the maximum index node that carries flow from {\em at least two} sources in $Y$.
\State let $Z_j\sse Y$ denote all remaining sources $s$ whose paths $\F(s)$ contain $r_j$.  
\State subtree $T_j$ consists of root node $r_j$, and for each source $s\in Z_j$, the $s-r_j$ prefix of path  $\F(s)$. 
 \State update $Y\gets Y\setminus Z_j$ and $j\gets j+1$.
\If{there is no root node (from $V'\setminus \{t\}$) satisfying the condition in step~\ref{step:single-clust-root}}\label{step:single-clust-nbr}
\For{each $u\in N$ (neighbor of $t$) }
\State set $r_j=u$ and $Z_j\sse Y$ is the singleton set containing the source whose path contains $u$. \Comment{{\em If there is no such source node then $Z_j=\emptyset$}}
\State subtree $T_j$ consists of root node $r_j$, and the $s-r_j$ prefix of path  $\F(s)$ for the source $s\in Z_j$.
\State update $Y\gets Y\setminus Z_j$ and $j\gets j+1$.  
\EndFor
\EndIf
\EndWhile 
\end{algorithmic}
\end{algorithm}

We first prove that the subtrees $T_j$ are node-disjoint.  
  Note that step~\ref{step:single-clust-nbr} in the while-loop only occurs when the sink $t$ is the only node containing flow from more than one source of $Y$. So, this can only happen in the last iteration.
  Consider any subtree $T_j$ produced in step~\ref{step:single-clust-root}. By the choice of root $r_j$, for each $s\in Z_j$  the portion of path $\F(s)$  from $s$ to $r_j$, is disjoint from the  paths of the remaining sources $Y\setminus Z_j$. That is, subtree $T_j$ is (node) disjoint from all subtrees $T_{j+1}, \cdots T_g$ found in later iterations. As noted above, step~\ref{step:single-clust-nbr}  only occurs in the last iteration, at which point every node in $V'\setminus \{t\}$ carries flow from at most one source of $Y$. So the subtrees produced in this step are also node-disjoint.

  It is clear that each source appears in some subtree, as the while-loop continues until $Y=\emptyset$: this proves property 1. It is also clear than each subtree $T_j\sse V'\setminus \{t\}$: combined with node-disjointness of the subtrees, we obtain  property 2.

We now bound the total demand in each subtree $T_j$. For any tree $T_j$ produced in step~\ref{step:single-clust-root}, we have  $r_j\in \F(s)$ for all $s\in Z_j$. So the  flow through node $r_j$ in the unsplittable flow $\F$ is at least $\sum_{s\in Z_j} d(s)$ the total demand in $T_j$. Using the fact that $\F$ sends at most $C+d_{max}$ flow through any node (other than $t$), the total demand in $T_j$ is at most $C+d_{max}$. As noted above, step~\ref{step:single-clust-nbr}  only occurs when every node in $V'\setminus \{t\}$ carries flow from at most one source of $Y$. So each subtree produced in step~\ref{step:single-clust-nbr}  contains at most one source, which has demand at most $d_{max}$. This proves property 3.
  
  Each subtree produced in step~\ref{step:single-clust-root} contains at least two sources: this follows from the choice of node $r_j$. All remaining subtrees (produced in step~\ref{step:single-clust-nbr}) have as their root some node of $N$ (neighbors of $t$). This proves property 4. 
  
For property 5, consider any subtree $T_j$ with root $r_j\not\in N$. Clearly, $T_j$ must be produced in step~\ref{step:single-clust-root}. Moreover, $T_j$ consists of the prefixes of certain paths until node $r_j$.  As nodes of $N$ only send flow to sink $t$ and the root $r_j\not\in N$, subtree $T_j$ does not contain any node of $N$.
\end{proof}


\section{Multicommodity Node-Capacitated Network Design}\label{sec:mc}

We now discuss the general multicommodity case of the problem. Recall that the input is an undirected
graph $G = (V,E)$ with $k$ request-pairs  $\{(s_i, t_i, d_i) \; | \; i \in [k]\}$, where the $i^{th}$ request has source $s_i$, sink $t_i$ and demand $1\le d_i\le q$. All nodes have capacity $q$. 
The output is a subset of nodes $V' \sse V$ such that the graph $G[V']$ induced by $V'$ can simultaneously support
$d_i$ units of flow (unsplittably) between nodes $s_i$ and $t_i$, for each $i\in[k]$.
The objective is to minimize the total cost $c(V') = \sum_{v \in V'} c_v$. As mentioned earlier, we assume  (without loss of generality) that all terminals are distinct and each terminal is a leaf node. 
For any terminal $s$, we define its \emph{mate} to be the unique terminal $t$ such that $(s,t)$ is a request-pair. We also use $d(s)$ to denote the demand associated with any terminal $s$; so we have $d(s_i)=d(t_i)=d_i$ for all $i\in [k]$.

\medskip {\bf Roadmap:} 
Our algorithm first clusters the terminals into nearly node-disjoint subtrees 
 of low total cost.  We need a new notion of ``allowed clusters'' in the multicommodity case, and the clustering algorithm is based on iteratively solving  several instances of the single-sink problem (\scnd). The details appear in  \S\ref{subsec:mc-cluster}. Next, the algorithm routes demands across different clusters while respecting node capacities. The routing algorithm relies on random-sampling and cut-sparsification: see \S\ref{subsec:mc-routing} for details. After combining the inter-cluster routing with the clusters themselves, we are able to route a {\em constant fraction} of the demands with small node congestion. Finally, we need to apply the above clustering and routing algorithms recursively on all unsatisfied demands: so we  repeat the main algorithm a logarithmic number of times.

\subsection{Clustering\label{subsec:mc-cluster}}
Here, we describe the multicommodity clustering algorithm that finds a collection of nearly disjoint clusters, where each cluster has either a large fraction of ``induced'' demands ({\em internal} clusters) or a large amount of ``crossing'' demands ({\em heavy} clusters). During our clustering algorithm, we will drop some request-pairs and maintain a {\em current} set of requests $K$. At any point in the algorithm, the terminals are the sources/sinks of {\em only} the current request-pairs.  We will ensure that requests remaining at the end of the clustering algorithm have a constant fraction of the total demand $D:=\sum_{i=1}^k d_i$.

\begin{definition}[\ncnd Cluster]
\label{def:mccluster} Let $K\sse[k]$ denote a subset of requests. The following definitions are relative to $K$, where the nodes $\{s_i,t_i\}_{i\in K}$ are called terminals.  
A cluster is any subtree $T$ in graph $G$. 
\begin{itemize}
\item The set of terminals contained in cluster $T$ is denoted $\asg(T)$. 
\item  The demand of cluster $T$ is $\load(T)=\sum_{s\in \asg(T)} d(s)$, the sum of demands over all its terminals. 
\item A terminal $s\in \asg(T)$ is called {\bf internal} if its mate is also in $\asg(T)$; the terminal $s$ is called {\bf external} otherwise. \item The internal (resp. external) demand of cluster $T$ is 
the total demand of its internal (resp. external) terminals.  
\end{itemize}
\end{definition}
Note that ``internal requests'' (with both source and sink  in $T$) contribute twice to $\load(T)$, whereas ``external requests'' (with exactly one terminal in $T$) contribute just once to $\load(T)$. 

\begin{definition}[Cluster Categories]\label{def:cluster-type} Let $K\sse[k]$ be a subset of requests and $T$ be any cluster. Then, $T$ is said to be 
\begin{itemize}
\item {\bf heavy} if its demand $\load(T)$ is at least $q$.
\item  {\bf internal}  if its internal demand is more than $\load(T)/2$.  
\item {\bf active} if it is neither internal nor heavy. 
\end{itemize}
\end{definition}

 We note that some clusters may be both internal and heavy. In our algorithm, we explicitly maintain collections of different cluster-types, and any ties will be broken according to the algorithm.

We will maintain and grow active clusters until all clusters are heavy or internal. We further classify active clusters into two types depending on how much of its external demand goes to other active clusters. This distinction is important because the algorithm needs to deal with these clusters differently.
\begin{definition}[Active Cluster Types]
\label{def:active-cluster}
Let $K\sse[k]$ be a subset of requests and  $T$ an active cluster. $T$ is a {\bf type 1} active cluster if the total demand of terminals in $T$ with their mates in other active clusters is less than $\load(T)/4$. Otherwise,  $T$ is  a {\bf type 2} active cluster. Moreover, a type 1 active cluster is called {\bf dangerous} if it has non-zero demand going to other active clusters. 
\end{definition}
We will refer to type 1 and type 2 active clusters as t1-active and t2-active, respectively. Note that the total external demand of  any active cluster $T$ is at least $\load(T)/2$: otherwise, $T$ would be an internal cluster (not active).  So, any t1-active cluster has at least $\load(T)/4$ demand crossing to internal/heavy clusters. Moreover, any t1-active cluster that is not dangerous has {\em all} 
its external demands going to internal/heavy  clusters.

We can now state the main multicommodity clustering result.

\begin{theorem} \label{thm:clustermc}
Suppose that there is a $(\costsingle, \congsingle)$-bicriteria approximation algorithm for the single-sink problem (\scnd). Then, for any \ncnd instance with optimal cost \opt, there is a polynomial-time algorithm that finds a subset $K\sse [k]$ of request-pairs and a collection ${\cal \widehat T}$ of  clusters (relative to $K$) such that:
\begin{enumerate}
\item[(i)] Each cluster in ${\cal \widehat T}$ is internal or heavy.
\item[(ii)] Each terminal (i.e. node in $\{s_i,t_i\}_{i\in K}$) lies in exactly one cluster.
\item[(iii)] The total demand of request-pairs $K$ is $\sum_{i\in K} d_i\ge D/4$.
\item[(iv)] Each node appears in  at most $O(\log  k)$ different clusters.
\item[(v)] The demand of each cluster is at most $O(\congsingle^2 \log k)\cdot q$.
\item[(vi)] The total cost of all the clusters 
$$\sum_{T\in {\cal \widehat T}}  \sum_{v\in T} c_v   \le   O(\costsingle \cdot \log k) \cdot \opt.$$
\end{enumerate}
\end{theorem}

\def\T{\ensuremath{{\cal T}}\xspace}
\def\N{\ensuremath{{\cal N}}\xspace}

\def\ti{\ensuremath{{\cal T}_i}\xspace}
\def\th{\ensuremath{{\cal T}_h}\xspace}
\def\tao{\ensuremath{{\cal A}_1}\xspace}
\def\tat{\ensuremath{{\cal A}_2}\xspace}
\def\mcc{\ensuremath{\mathsf{MCcluster}}\xspace}
\def\mto{\ensuremath{\mathsf{MergeT1}}\xspace}
\def\mtt{\ensuremath{\mathsf{MergeT2}}\xspace}
\def\cop{\ensuremath{9\congsingle}\xspace}

\def\tf{\ensuremath{{\cal T}_f}\xspace}
\def\ts{\ensuremath{{\cal T}_s}\xspace}
\def\tu{\ensuremath{{\cal T}_u}\xspace}

\def\tux{\ensuremath{\tu^\times}\xspace}

\medskip
\noindent {\bf Overview of algorithm/analysis.}  
We start with each terminal being its own cluster. The clustering algorithm aims to find clusters with $\approx q$ terminals in each: these aggregated demands can then be handled in the routing step of our algorithm. This motivates the definition of {\em heavy} clusters, which contain at least $q$ demand (see Definition~\ref{def:cluster-type}). In order to obtain such heavy clusters, the algorithm iteratively merges the non-heavy clusters. Moreover, to ensure that there is a low-cost solution to this ``merging'' step, we need each cluster to have a constant fraction of its demand going to {\em other} clusters: otherwise, the cost to merge may be much more than the optimal \ncnd cost (denoted by \opt). This motivates the definition of {\em internal} clusters, where a constant fraction of the demand is induced inside the cluster (see Definition~\ref{def:cluster-type}). While internal clusters can not participate in  merging anymore, we can use the subtree inside such clusters to route all its internal demand (which is a constant fraction of its total demand). So, both heavy and internal clusters are ``good'' in the sense that we will be able to satisfy a constant fraction of their demand in the subsequent routing step.  Therefore, the revised aim of the clustering algorithm is to find a collection of heavy {\em or} internal clusters. All other clusters are called {\em active}, which the algorithm continues to merge.

The clustering algorithm proceeds in iterations. Each iteration attempts to merge active clusters, and ensures that the number of active clusters reduces by a constant factor. So, the number of iterations will be at most $O(\log k)$. 
The merging step in each iteration is based on solving suitable instances of the {\em single-sink} problem \scnd. We will ensure  that the optimal value of each \scnd instance is at most \opt: so the final cost of our clusters will be $O(\log k)\cdot \opt$.  
When multiple active clusters merge together in any iteration, the resulting cluster  may be internal, heavy or active (we make progress in all cases). However,  the new cluster may have demand  much more than $q$, as we cannot control how many active clusters get merged into one. In order to handle this issue, we ensure that the \scnd instances have node-capacity $\approx q$: the \scnd approximation guarantee then implies that the demand in any new cluster is at most $\approx \congsingle\cdot q$. (In some cases  we will have slightly larger node capacities in \scnd, as explained later).  

Another issue to handle is that we may be unable to merge active clusters just with each other. In particular, it is possible that a large fraction of an active cluster's demand goes to heavy/internal clusters. Then, we cannot expect to merge such a cluster with other active clusters. This motivates the classification of active clusters into types 1 and 2, corresponding to having low/high fraction of its demand to other active clusters (see Definition~\ref{def:active-cluster}). The two types of active clusters are dealt with separately. Initially, each singleton cluster is t2-active. Assume for now that there are no requests between active clusters of different types. We will explicitly ensure this property by {\em dropping some requests} and preserving only a subset $K\sse [k]$ of requests.  
\begin{itemize}
\item {\em Merging t2-active clusters.} Intuitively, these active clusters can be merged with each other (at low cost) because most of their demands go to other  active clusters. We would like to construct an \scnd instance ${\cal I}_2$, where each t2-active cluster is a source. 
 However, the original \ncnd requests are multicommodity; so, even requests associated with one cluster do not have a common sink.  We get around this issue by restricting attention to a ``bipartite demand graph'' containing at least half the total demand between t2-active clusters. This corresponds to finding an appropriate bi-partition 
$({\cal A}^+, {\cal A}^-)$ of t2-active clusters, which can be done using  a simple local search. Now, we  treat all clusters in 
part ${\cal A}^+$ as sources and connect all clusters in ${\cal A}^-$ to a new sink node $t$. All nodes have capacity $\approx q$. See Algorithm~\ref{alg:t2-merge} for the formal description. In the analysis, we need to show that the optimal cost of this \scnd instance is at most \opt. This is done by demonstrating a   (fractional) routing 
based on the optimal \ncnd solution and using the fact that at least half the total demand in t2-active clusters are ``crossing'' between ${\cal A}^+$ and ${\cal A}^-$.  Finally, we obtain a collection of subtrees (using the  approximate \scnd solution and Theorem~\ref{thm:ss-disj-cluster}) that are used to merge t2-active clusters.

\item {\em Merging  t1-active clusters.} These clusters have a constant fraction of their demands going to heavy/internal clusters. If any requests from a t1-active cluster go   to an internal cluster  then we simply drop these requests and ``charge'' them to the requests in internal clusters  which will definitely be preserved. Note that demand induced in any internal cluster is at least a constant fraction of its total demand. Moreover, the clustering algorithm only aims to preserve  a subset $K$ of requests (which should be a constant fraction of the total demand). 
So,  we are left with the case that a 
large fraction of demand from any t1-active cluster goes to heavy clusters. Intuitively, these clusters can be merged with heavy clusters  at low cost.  We now construct an \scnd instance ${\cal I}_1$  to merge t1-active clusters with heavy clusters (or each other). We  treat each t1-active cluster as a source, and connect all heavy clusters to a new sink $t$. All nodes have capacity $\approx q$ except the nodes corresponding to heavy clusters, which have capacity $O(\congsingle\cdot q)$. The reason that we  have  a larger capacity  for heavy clusters is that the demand in a heavy cluster can be as large as $O(\congsingle\cdot q)$ and the \ncnd routing (which is  used to demonstrate a low-cost \scnd routing) induces a corresponding load on these clusters. 
See Algorithm~\ref{alg:t1-merge} for the formal description. Again, we merge clusters 
using the  approximate \scnd solution and Theorem~\ref{thm:ss-disj-cluster}. \end{itemize}

The demand of any new cluster formed when active clusters merge with each other is bounded by $O(\congsingle \cdot q)$ as all nodes in these clusters have capacity $O(q)$ and $\gamma$ is the capacity violation in our \scnd algorithm. However, when t1-active clusters merge with an existing heavy cluster $H$, the demand of $H$ may increase by as much as $O(\congsingle^2 \cdot q)$: this is because nodes corresponding to heavy clusters  have capacity $O(\congsingle\cdot q)$. Unfortunately, this increased demand in $H$ may multiply over iterations. Specifically, if most of the new requests in $H$ are going to other active clusters then the \scnd instance ${\cal I}_1$ in the next iteration must increase the capacity of $H$ from $\congsingle\cdot  q$ to  $ \congsingle^2\cdot   q$; this is needed to demonstrate a low cost solution to ${\cal I}_1$. So, the capacity (and hence the demand) of cluster $H$ will keep increasing by a multiplicative factor $\congsingle$ in each iteration! In order to fix this issue, we will ensure that all t1-active clusters in \scnd instance ${\cal I}_1$ have {\em zero} demand going to other active clusters (again, this property will be ensured by dropping certain requests). This   motivates the definition of {\em dangerous} clusters, which are t1-active clusters having any request going to other active clusters (Definition~\ref{def:active-cluster}). We will ensure that instance  
${\cal I}_1$  has no dangerous clusters. Now, when t1-active clusters merge with a heavy cluster  $H$, the demand between $H$ and active clusters {\em does not} increase. So, the capacity of  $H$ in instance ${\cal I}_1$ of the next iteration can remain  $O(\congsingle\cdot  q)$. We note that  the total demand in $H$ still increases {\em additively} by $O(\congsingle^2\cdot q)$ in each iteration. As there are only $O(\log k)$ iterations, the final demand of any heavy cluster can be bounded by $O(\log k \cdot \congsingle^2\cdot  q)$. 

{\em Dropping requests.} In the above description, we assumed that there are no requests between (i) active clusters of different types, (ii) internal clusters and active clusters, and (iii) t1-active clusters and t1- or t2-active clusters. As mentioned earlier,  to ensure these properties, our algorithm drops certain requests  and preserves only a subset $K\sse [k]$. Specifically, in each iteration  we perform  a ``pruning step'' after merging active clusters. This involves repeatedly choosing a (new) cluster $T$ that is either internal or dangerous, and dropping all requests between $T$ and other active clusters. In the analysis we will show that the dropped requests can be ``charged'' to preserved requests in $K$, so that the final set $K$ is a constant fraction of the total demand in \ncnd. No request incident to a heavy cluster is ever dropped: so heavy clusters do not change in this pruning step.  We note that dropping requests incident to cluster $T$ may cause another active cluster $T'$ to become internal or dangerous: then, cluster $T'$ will also be processed later in this pruning step. Therefore, at the end of each iteration, we ensure the three properties (i)-(iii). We emphasize that the cluster definitions are all relative to the current set $K$ of requests.

\medskip
Algorithm~\ref{alg:cluster-mc} describes the overall clustering algorithm \mcc. We maintain separate collections of clusters: \ti (internal), \th (heavy), \tao (t1-active) and \tat (t2-active). Algorithm~\ref{alg:t1-merge} (\mto) and Algorithm~\ref{alg:t2-merge} (\mtt) describe the separate procedures to merge t1-active and t2-active clusters. 
 Figures~\ref{fig:G1} and \ref{fig:G2} illustrate the graphs $G_1$ and $G_2$ used in the two \scnd instances ${\cal I}_1$ (solved in \mto) and  
${\cal I}_2$ (solved in \mtt).

\begin{algorithm}[h]
\caption {\ncnd Clustering Algorithm \mcc \label{alg:cluster-mc}}
\begin{algorithmic}[1]
\State initialize $K = [k]$ to be all request-pairs and the clusters are all singletons.  
\State all clusters are t2-active, i.e., $\tat \gets \{\{s_i\},\{t_i\}\}_{i\in K}$, $\ti\gets \emptyset$, $\th\gets \emptyset$, $\tao\gets \emptyset$.
\While{some cluster is active ($\tao\cup \tat\ne \emptyset$)}\label{step:cluster-iter1}\Comment{{\em Iteration begins}}
\State   run algorithm $\mto(\th, \tao)$ which modifies t1-active and heavy clusters.
\State   run algorithm $\mtt(\tat)$ which modifies t2-active  clusters.
\For{all clusters $T$ in $\tao\cup \tat$} \Comment{{\em Identify heavy clusters}}
\State  if $T$ is heavy then move it from $\tao$ or $\tat$ to $\th$. \label{step:cluster-heavy}
\EndFor 
\While{some cluster $T$ in $\tao\cup \tat$ is dangerous or internal }  \label{step:cluster-prune1} 
\Statex \Comment{{\em Identify internal clusters and ensure no dangerous clusters}}
\State   \label{step:cluster-remove-req} remove  from $K$ all requests between $T$ and other active clusters.
\State  if $T$ is internal (resp. t1-active) then move it  to $\ti$ (resp. \tao).  
\EndWhile \label{step:cluster-prune2} 
\State  if any cluster $T\in \tat$ is t1-active then move it to \tao. \label{step:cluster-prune3}  
\EndWhile \label{step:cluster-iter2}  
\end{algorithmic}
\end{algorithm}

\begin{algorithm}[h!]
\caption {Merging algorithm for t1-active clusters $\mto(\th, \tao)$ \label{alg:t1-merge}}
\begin{algorithmic}[1]
\State let graph $G_1$ consist of graph $G$ and the following new nodes/edges:
\begin{itemize}
\item For each cluster $T\in \tao$,
there is a source node $s_T$ (of zero cost) with demand $d(s_T) = \load(T)$; node $s_T$ is connected to each terminal in $T$. 
\item For each heavy cluster $F\in \th$, there is a new node $v_F$ (of zero cost), which is connected to every terminal in $F$.
\item There is a new sink node $t$, connected to the nodes $\{v_F:F\in \th\}$.
 \end{itemize} 
\State The node capacities are $\tilde{q}:=5q$. For the zero-cost nodes $v_F$ corresponding to heavy clusters $F\in \th$, we set their capacities to $\cop\cdot \tilde{q}$, where $\congsingle$ is the approximation ratio for node-congestion from our single-sink algorithm. 
\State  let ${\cal I}_1$ be the \scnd instance on graph $G_1$, with sources $\{s_T\}_{T\in \tao}$ and sink $t$.  
\State  solve instance ${\cal I}_1$  using the \scnd approximation algorithm (Theorem~\ref{thm:single-sink}) to obtain solution $V_1\sse V$.
\State  using Theorem~\ref{thm:ss-disj-cluster} on solution $V_1$,  obtain a collection of rooted subtrees ${\cal N}_1$.
\For{each subtree $X$ (with root $r$)  in ${\cal N}_1$} \Comment{{\em Merge clusters using ${\cal N}_1$}} 
\State  let ${\cal X} =\{T\in \tao: s_T\in X\}$ be the t1-active clusters whose source-nodes are contained in  subtree $X$ 
\State  update $\tao\gets \tao\setminus {\cal X}$
\If{root $r$ corresponds to a heavy cluster $F\in \th$ (i.e., $r=v_F$)}
\State  \label{step:heavy-merge} add clusters ${\cal X}$ and subtree $X$ to heavy cluster $F$, i.e. $F\gets F\cup X\cup {\cal X}$ 
\Else
\State \label{step:t1-merge} merge  clusters ${\cal X}$ using subtree $X$ to get new cluster $S=X\cup {\cal X}$
\State  update $\tao\gets \tao\cup \{S\}$ 
\Statex \Comment{{\em $S$ may not be t1-active: its  type will be updated in \mcc}}
\EndIf
\EndFor   
\end{algorithmic}
\end{algorithm}
\vspace{-4mm}

\begin{algorithm}[h!]
\caption {Merging algorithm for t2-active clusters $\mtt(\tat)$  \label{alg:t2-merge}}
\begin{algorithmic}[1] 
\State \label{step:t2-partn} partition clusters $\tat$ into ${\cal A}^+$ and ${\cal A}^-$ such that every cluster in ${\cal A}^+$ (resp. ${\cal A}^-$) has more demand going to clusters in ${\cal A}^-$ (resp. ${\cal A}^+$) than ${\cal A}^+$ (resp. ${\cal A}^-$)
\Statex \Comment{{\em This can be done by a simple local search}}
\State relabel the parts so that $|{\cal A}^+| \ge |{\cal A}^-|$.
\State let graph $G_2$ consist of graph $G$ and the following new nodes/edges:
\begin{itemize}
\item For each cluster $T\in {\cal A}^+$,
there is a  source node $s_T$ (of zero cost)  with demand $d(s_T) = \load(T)$; node $s_T$ is connected to  each terminal of $T$.
\item For each cluster $W\in {\cal A}^-$, there is a new node $v_W$ (of zero cost), connected to every terminal of $W$.
\item There is a new sink node $t$, connected to the nodes $\{v_W:W\in {\cal A}^-\}$.
\end{itemize}
\State The node capacities are $q' =9q$.
\State let ${\cal I}_2$ be the \scnd instance on graph $G_2$, with sources $\{s_T\}_{T\in {\cal A}^+}$ and sink $t$.  
 \State  solve instance ${\cal I}_2$  using the \scnd approximation algorithm (Theorem~\ref{thm:single-sink}) to obtain solution $V_2\sse V$.
\State using Theorem~\ref{thm:ss-disj-cluster} on solution $V_2$,  obtain a collection of rooted subtrees ${\cal N}_2$.
\For{each subtree $Y$  in ${\cal N}_2$} \Comment{{\em Merge clusters using ${\cal N}_2$}} 
\State  let ${\cal Y} =\{T\in {\cal A}^+ : s_T\in Y\}\bigcup \{W\in {\cal A}^- : v_W\in Y\} $ be the t2-active clusters whose $s$-nodes or $v$-nodes are contained in  subtree $Y$ 
\State update $\tat\gets \tat\setminus {\cal Y}$.
\State \label{step:t2-merge}  merge  clusters ${\cal Y}$ with each other to get new cluster $S=Y\cup {\cal Y} $
\State update $\tat\gets \tat\cup \{S\}$ \Comment{{\em $S$'s type will be updated in \mcc}}
\EndFor   
\end{algorithmic}
\end{algorithm}

\medskip

\noindent {\bf Analysis.}  
Our first lemma shows that all clusters are classified correctly (relative to the current requests $K$) at the end of each iteration of \mcc. During the iteration, new or modified  clusters may be in the wrong  collection: but these will get fixed at the end (as shown in the next lemma).  
\begin{lemma}\label{lem:mc-clust-class}
At the end of each iteration of Algorithm \mcc, the current set of clusters are correctly classified into 
\ti (internal), \th (heavy), \tao (t1-active) and \tat (t2-active). 
\end{lemma}
\begin{proof}
At the beginning of the algorithm, each cluster is a singleton terminal. Clearly, these are t2-active clusters: so the initial classification is correct.

Algorithm \mcc first runs \mto and \mtt to merge t1-active and t2-active clusters separately. 
\mtt creates new clusters by merging t2-active clusters, and places all new clusters in \tat. \mto modifies existing heavy clusters (which remain in \th) and creates new clusters by merging t1-active clusters, which are placed in \tao.  Then, in Step~\ref{step:cluster-heavy} of \mcc,  we first identify any new heavy clusters (in $\tao\cup\tat$) and move them to \th. At this point, all heavy clusters are classified correctly. Moreover, \mcc never removes requests incident to a  heavy cluster (see Steps~\ref{step:cluster-prune1}-\ref{step:cluster-prune2}). So, once a cluster is heavy, it will remain heavy throughout the algorithm. 

Next, in Steps~\ref{step:cluster-prune1}-\ref{step:cluster-prune2} of \mcc, we repeatedly identify dangerous/internal clusters and  drop some requests (which in turn can modify other clusters).  We now argue that the clusters are classified correctly after this step:
\begin{itemize}
\item Suppose an  internal cluster $T\in \tao\cup\tat$ is processed in Step~\ref{step:cluster-prune1}. We remove all requests from $T$ to other active clusters. This keeps $T$ as an internal cluster, and it is classified correctly. Moreover, we never remove any internal requests, and we never modify an internal cluster after it is assigned to \ti. 
\item Suppose a dangerous cluster $T\in \tao\cup\tat$ is processed in Step~\ref{step:cluster-prune1}. We remove all requests from $T$ to other active clusters.   
As a result, $T$ will become either internal (if its internal demand is more than $\load(T)/2$) or  non-dangerous t1-active (if its internal demand is at most $\load(T)/2$); note that $T$ cannot become heavy or t2-active. In either case, $T$ is classified correctly and it will not be modified later in this iteration.
\item Consider now any cluster $T'\in \tao$ that is {\em not} processed in Step~\ref{step:cluster-prune1}. Then,  $T'$ is not heavy or internal  (or even dangerous). So it must be active. Also, it must have formed in algorithm \mto as a result of merging some t1-active clusters and such a cluster cannot be t2-active: its demand going to other active clusters will be less than $\load(T')/4$. Further, when requests are dropped, a t1-active cluster cannot become t2-active because the only requests dropped from $T'$ are those going to other active clusters. So $T'$ remains t1-active at the end of this iteration.
\item Consider now any cluster $T'\in \tat$ that is {\em not} processed in Step~\ref{step:cluster-prune1}. Again,  $T'$ must be active and non dangerous. It could be either t1-active or t2-active, and is classified correctly in Step~\ref{step:cluster-prune3}.
\end{itemize}
\end{proof}

\begin{figure}[h]
    \centering
    \begin{minipage}{.55\textwidth}
        \centering
        \includegraphics[width=0.9\linewidth, height=0.25\textheight]{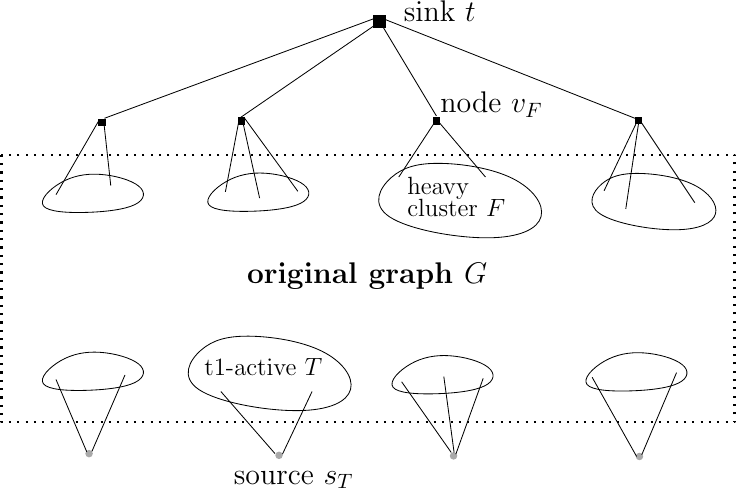}
        \caption{Graph $G_1$ for t1-active clusters.}
        \label{fig:G1}
    \end{minipage}%
    \begin{minipage}{0.45\textwidth}
        \centering
        \includegraphics[width=0.9\linewidth, height=0.25\textheight]{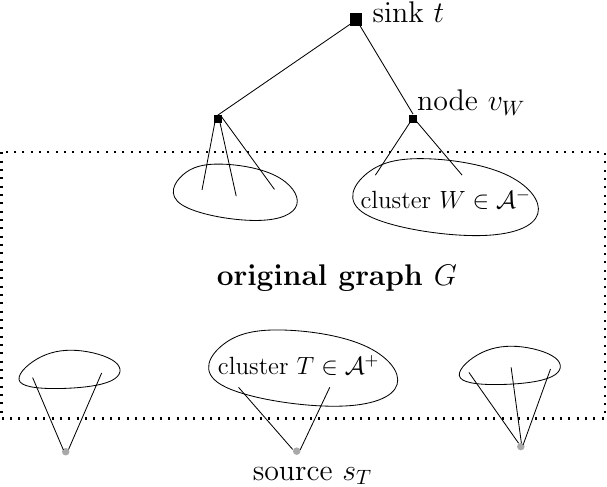}
        \caption{Graph $G_2$ for t2-active clusters.}
        \label{fig:G2}
    \end{minipage}
\end{figure}

We now observe how cluster types change during algorithm \mcc. 
The classification of a cluster (its assignment to \th, \ti, \tao or \tat) changes due to  merging the cluster with other clusters (which occurs in algorithms \mto and \mtt), or  removal of requests from $K$ (which occurs in algorithm \mcc).  
We note that these changes satisfy following:
\begin{itemize}
\item 
Any cluster in  \ti  will remain in  \ti.   
\item 
Any cluster in \th will be (part of) some  cluster of \th.  
\item Any cluster in \tao will be (part of) some  cluster of  $\tao\cup \ti\cup\th$.
\item Any cluster in \tat will be (part of) some  cluster of $\tat\cup \tao\cup\ti\cup\th$.
\end{itemize}

Next, we  show that requests (in the current set $K$) between different kinds of clusters satisfy some useful properties. These properties are crucial in setting up the \scnd instances correctly and ensuring that the demand of heavy clusters does not grow too much.  Figure~\ref{fig:mc-cluster} illustrates the possible requests across clusters.

\begin{figure}[h]
\centering
\includegraphics[scale=0.99 ]{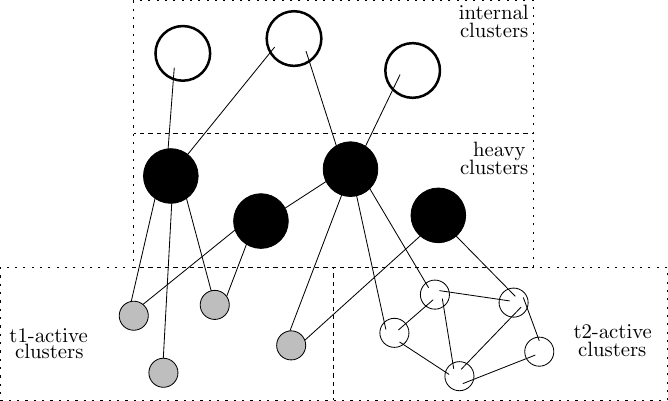}
\caption{Different types of clusters at the end of an \mcc iteration. The edges represent request-pairs (in $K$) across clusters.  \label{fig:mc-cluster}}
\end{figure}

\begin{lemma}\label{lem:mc-cluster-structure} At the start/end of any iteration of algorithm \mcc, there are no requests (in $K$) between:
\begin{itemize}
\item internal clusters \ti and active clusters $\tao\cup \tat$.
\item any t1-active cluster $T\in\tao$ and other active clusters $\tao\cup \tat \setminus \{T\}$.
\end{itemize}
Moreover, the demand from any heavy cluster to active clusters is at most $10\congsingle\cdot q$. 
\end{lemma}
\begin{proof}
We will prove this inductively over the iterations. The lemma is clearly true at the beginning of \mcc (all clusters are t2-active).
Now consider any iteration of \mcc: assuming the lemma at the start of the iteration, we prove that it also holds at the end.  

For the first property, observe that \mcc does not modify any existing internal cluster: so any cluster that was internal at the start of the iteration continues to satisfy this property. Let $I$ be a (new) cluster that is found to be internal in this iteration. 
This must happen in Steps~\ref{step:cluster-prune1}-\ref{step:cluster-prune2}, and we explicitly remove all requests from $I$ to other active clusters at this point.

For the second property, let $T\in \tao$ be any t1-active cluster at the end of the iteration. By the loop condition in  Steps~\ref{step:cluster-prune1}-\ref{step:cluster-prune2} of \mcc, $T$ cannot be dangerous. So, the demand from $T$ to other active clusters is zero, as desired.

For the third property, we consider two cases for any heavy cluster $H\in \th$:
\begin{itemize}
\item {\em $H$ is a new  heavy cluster.} Then, $H$ must have formed due to merging t1-active (resp. t2-active) clusters in algorithm \mto (resp. \mtt). See Step~\ref{step:t1-merge} in \mto and Step~\ref{step:t2-merge} in \mtt. If $H$ was formed in \mto then it was based on \scnd instance ${\cal I}_1$, where node capacities are $\tilde{q}=5q$. Using Theorems~\ref{thm:single-sink} and \ref{thm:ss-disj-cluster}, it follows that the demand of cluster $H$ must be at most $\congsingle\cdot \tilde{q} +  \tilde{q} \le 10\congsingle\cdot q$. If $H$ was formed in \mtt, the analysis is similar-- this time, using \scnd instance  ${\cal I}_2$ which has node capacity $q'=9q$.  Using Theorems~\ref{thm:single-sink} and \ref{thm:ss-disj-cluster} again, the demand of cluster $H$ is at most $\congsingle\cdot q' +q' = (9\congsingle + 9)\cdot q\le 10 \congsingle\cdot q$. Here, we used $\congsingle\ge 9$ because $\congsingle$ is poly-logarithmic.  In either case, the total demand in cluster $H$ is at most $10\congsingle\cdot q$, which also bounds the demand from $H$ to active clusters. 
\item {\em $H$ is an existing  heavy cluster. } Then, $H$ will be modified in algorithm \mto based on \scnd instance ${\cal I}_1$. Here, only t1-active clusters get added to cluster $H$: see Step~\ref{step:heavy-merge} in \mto. Crucially, there is {\em zero} demand from any t1-active cluster to other active clusters (by the second property in this lemma). So the  demand from cluster $H$ to  active clusters does not increase, and it remains at most  $10\congsingle\cdot q$. 
\end{itemize}
\end{proof}

We are now ready to bound the  optimal values of the \scnd instances ${\cal I}_1$ and ${\cal I}_1$. We do this by using the optimal \ncnd solution and properties of the different cluster-types (shown above).

\begin{lemma} \label{lem:lowcostmc1}
The optimal cost of \scnd instance ${\cal I}_1$ (in algorithm \mto) is at most the optimal cost $\opt$ of the original \ncnd instance.
\end{lemma}
\begin{proof}
We first exhibit a feasible {\em fractional} flow  using the optimal solution of the multicommodity  \ncnd instance. Let $V^*\sse V$ be the nodes used in \opt. For each request $i\in [k]$, let $P^*_i$ denote the path from $s_i$ to $t_i$ in \opt. Note that $P^*_i\sse V^*$ for all $i\in[k]$, and $\sum_{i\in [k] : P^*_i\ni v} d_i \le q$ for all nodes $v\in V^*$. 

We will use $V^*\cup \{s_T: T\in \tao\}\cup \{v_F: F\in \th\}\cup \{t\}$ as the nodes in our \scnd solution. The cost of this solution is $\opt$ as the new nodes have zero cost. We now show how to route all the demands fractionally using these nodes.

Consider any t1-active cluster $T\in \tao$. We first claim that the total demand of  terminals in $T$ having mates in heavy clusters is at least $\load(T)/4$. Indeed, by definition of t1-active clusters, the total demand from $T$ to internal/heavy clusters is at least $\load(T)/4$. Moreover, by Lemma~\ref{lem:mc-cluster-structure}, there is zero demand from $T$ to internal clusters: so the demand from $T$ to heavy clusters $\th$ is at least $\load(T)/4$.  We now route  demand   from $s_T$ to $t$ as follows. For each external terminal $s_i \in T$ (with its mate $t_i\in F$ for some heavy cluster $F$), send $4$ units of flow from $s_T$ to $s_i$, then from $s_i$ to $t_i$ along path $P^*_i$, then from $t_i$ to $v_F$, and finally from $v_F$ to sink $t$. Note that these are valid paths in graph $G_1$ of instance ${\cal I}_1$. 
Also, the total demand routed from $s_T$ to $t$ is at least $\load(T)=d(s_T)$, as desired.  

Routing as above for each cluster $T\in\tao$, we get a fractional routing that satisfies all the demand in the \scnd instance ${\cal I}_1$.  We now argue that the load on any node is at most its capacity. Clearly, the flow through each node $v\in V^*$ is at most $4\cdot \sum_{i\in [k] : P^*_i\ni v} d_i  \le 4q$. Moreover, the flow through each node $v_F$ (for $F\in \th$) is at most $4$ times the total demand between $F$ and all active clusters, which is at most $40\congsingle q$ by Lemma~\ref{lem:mc-cluster-structure}. 

Let $d_{max} = \max_{T\in\tao} \load(T)$ denote the maximum demand in ${\cal I}_1$; note that $d_{max}\le q$.  Now, we convert the above fractional flow into an unsplittable flow as required in \scnd. To this end,
 we use Theorem~\ref{thm:unsplit} on our fractional routing for ${\cal I}_1$. We then obtain \emph{an unsplittable flow} for ${\cal I}_1$, where (i) the flow through each node of $V^*$ is at most $4q+d_{max}\le 5q=\tilde{q}$, and (ii) the flow through each node $v_F$ is at most $40\congsingle q + d_{max}\le \cop \tilde{q}$. 
 The lemma now follows. \end{proof}

\begin{lemma} \label{lem:lowcostmc2}
The optimal cost of \scnd instance ${\cal I}_2$ (in algorithm \mtt) is at most the optimal cost $\opt$ of the original \ncnd instance.
\end{lemma}
\begin{proof}
The high-level proof is similar to that of Lemma~\ref{lem:lowcostmc1}. We first show a feasible {\em fractional} routing using the \ncnd optimal solution, and then convert it into an unsplittable flow. Let $V^*\sse V$ be the nodes used in \opt. For each request $i\in [k]$, let $P^*_i$ denote the path from $s_i$ to $t_i$ in \opt. Again, $P^*_i\sse V^*$ for all $i\in[k]$, and $\sum_{i\in [k] : P^*_i\ni v} d_i \le q$ for all nodes $v\in V^*$.

We refer to the clusters in ${\cal A}^+$ as source clusters as they correspond to source nodes in the \scnd instance 
${\cal I}_2$. We also refer to the clusters in  ${\cal A}^-$ as sink clusters as they are directly connected to the sink $t$.  
Note that all these clusters are t2-active. 

Consider any source cluster $T\in {\cal A}^+$. We first show that 
the total demand from $T$ to sink clusters ${\cal A}^-$ is at least $\load(T)/8$.
 By definition of t2-active clusters, the total demand from $T$ to other active clusters is at least $\load(T)/4$. Moreover, by Lemma~\ref{lem:mc-cluster-structure}, there are no requests between $T$ and t1-active clusters. So, the total demand from $T$ to other t2-active clusters is at least $\load(T)/4$. Further, by choice of the partition $({\cal A}^+, {\cal A}^-)$ in Step~\ref{step:t2-partn} of \mtt, the total demand from 
any source cluster $T$ to all sink clusters is at least half the total demand from $T$ to t2-active clusters. So, there is demand at least $\load(T)/8$ from $T$ to  ${\cal A}^-$. 

Now, for any source cluster $T$,  let $C_T$ denote the set of all requests between $T$ and sink clusters. We route demand from $s_T$ 
to $t$ as follows. For each $i\in C_T$, send $8$ units from $s_T$ to $s_i$, then from $s_i$ to $t_i$ along path $P^*_i$, then from $t_i$ to $v_W$ (where $t_i$ lies in sink cluster $W\in {\cal A}^-$), and finally from $v_W$ to $t$. Note that these are valid paths in graph $G_2$ of instance ${\cal I}_2$. Moreover, the net flow out of each source $s_T$ is at least $\load(T)$ as desired.

Performing the above routing for all source clusters $T\in{\cal A}^+$, we obtain a fractional flow that satisfies all demands in ${\cal I}_2$. We now argue that the node capacity constraints are satisfied. Clearly, the flow through each node $v\in V^*$ is at most $8\cdot \sum_{i\in [k] : P^*_i\ni v} d_i  \le 8q$. Moreover, the flow through each node $v_W$ is at most $8$ times the total external demand in cluster $W$, which is at most $8\cdot q$ because $W$ is an active cluster.
Now, applying Theorem~\ref{thm:unsplit}, we obtain an unsplittable flow for ${\cal I}_2$, where the flow through each node is at most $8q+\max_{T\in{\cal A}^+} \load(T)\le 9q=q'$. This completes the proof.
\end{proof}

We now summarize some key properties of the (partial) solution built in each iteration of algorithm \mcc.

\begin{lemma}\label{lem:mcc-iteration} In  any iteration of \mcc, we have: 
 \begin{enumerate}
\item  The subtrees added to clusters  are ${\cal N}_1\cup {\cal N}_2$, where ${\cal N}_1$ and ${\cal N}_2$ are found in \mto and \mtt. Each node appears in at most two of these subtrees. 
\item  The  cost of the new nodes added to clusters is at most $2\costsingle\cdot \opt$.
\item The demand of any new internal/heavy cluster 
 is at most $10\congsingle\cdot q$. 
 \item The demand of any existing heavy cluster  increases by at most $54\congsingle^2\cdot q$.
\item  The number of active clusters at the end of the iteration is at most $\frac34$ times  the number at the start of the iteration.
\end{enumerate}
\end{lemma}
\begin{proof}
We prove the claimed properties one by one.

{\em Property 1.} 
In any iteration, subtrees are added to clusters in both \mto and \mtt. The subtrees added in \mto are exactly those in ${\cal N}_1$, which is the collection obtained  by applying Theorem~\ref{thm:ss-disj-cluster} to the \scnd solution $V_1$. Similarly, the  subtrees added in \mtt are exactly those in ${\cal N}_2$. By Theorem~\ref{thm:ss-disj-cluster}, all subtrees in   ${\cal N}_1$ (resp.  ${\cal N}_2$) are node-disjoint. Hence, each node appears at most twice in ${\cal N}_1\cup {\cal N}_2$. 

{\em Property 2.} By  Theorem~\ref{thm:ss-disj-cluster} (property 2), the total cost of subtrees in ${\cal N}_1$ is at most $c(V_1)$. Moreover, by our \scnd approximation guarantee, $c(V_1)$ is at most $\costsingle$ times the optimal value of instance ${\cal I}_1$. Using Lemma~\ref{lem:lowcostmc1}, it now follows that $c({\cal N}_1) \le c(V_1)\le \costsingle\cdot \opt$. Similarly, we obtain $c({\cal N}_2) \le c(V_2)\le \costsingle\cdot \opt$ using Lemma~\ref{lem:lowcostmc2} (for \scnd instance ${\cal I}_2$). Hence, the total cost of the new nodes added to clusters is at most $2\costsingle\cdot \opt$. 

{\em Property 3.} In any iteration, new clusters may be formed in either \mto or \mtt. Consider any new cluster formed in \mtt: see Step~\ref{step:t2-merge}. Note that this cluster corresponds to some subtree $Y\in {\cal N}_2$. So, by property 3 in Theorem~\ref{thm:ss-disj-cluster}, its total demand is at most $C' + d_{max}({\cal I}_2)$ where $C'$ is the maximum flow through any node in our \scnd solution for ${\cal I}_2$. By our \scnd approximation guarantee and the fact that all node capacities are $q'=9q$ (in ${\cal I}_2$), we have $C'\le \congsingle \cdot q'$. Also, the maximum demand $d_{max}({\cal I}_2)\le q$ as each source node corresponds to an active cluster. So, the total demand of the new cluster is at most $(9  \congsingle +1)q$. Now, consider a new cluster $S$ formed in \mto: see Step~\ref{step:t1-merge}. This cluster corresponds to some subtree $X\in {\cal N}_1$ with root $r\not\in N=\{v_F:F\in \th\}$; note that $N$ is the set of neighbors of sink $t$ in instance ${\cal I}_1$. Again, by property 3 in Theorem~\ref{thm:ss-disj-cluster}, the demand of $S$ is at most $C'' + d_{max}({\cal I}_1)$ where $C''$ is the maximum flow through any node in our \scnd solution for ${\cal I}_1$. By our \scnd approximation guarantee and the fact that node capacities in ${\cal I}_1$ are $\tilde{q}=5q$, we have $C''\le \congsingle \cdot \tilde{q}$. (By property 5 of Theorem~\ref{thm:ss-disj-cluster}, cluster $S$ does not contain any node of $N$: so it is not affected by the larger capacity on the $v$-nodes.) Again, the maximum demand $d_{max}({\cal I}_1)\le q$. So, the total demand of the new cluster is at most $(5  \congsingle +1)q$. 

{\em Property 4.} Consider any existing heavy cluster $F\in \th$. The clusters that get added to $F$ correspond to subtrees $X\in {\cal N}_1$ with root $r=v_F$; see Step~\ref{step:heavy-merge} in \mto. Moreover, node $v_F$ has capacity $9\congsingle$ (which is equivalent to having $9\congsingle$ {\em copies} of $v_F$). So, there may be up to   $9\congsingle$ subtrees $X\in {\cal N}_1$ with root $v_F$. Each such subtree has demand at most $(5  \congsingle +1)q$, as shown above. Hence, the increase in demand of $F$ is at most $9\congsingle(5  \congsingle +1)q\le 54\congsingle^2 q$.  

{\em Property 5.} Let $m_1$  (resp. $m_2$)   be the number of t1-active (resp. t2-active) clusters at the start of the iteration.  Let $\tao'$ (resp. $\tat'$)   be the  t1-active (resp. t2-active) clusters at the end of algorithm \mto (resp. \mtt). 
Any cluster that is active at the end of the iteration must be in $\tao'\cup \tat'$. (Note that some clusters in  $\tao'\cup \tat'$ may become internal/heavy during the pruning step in \mcc.) We now bound $|\tao'|$ and $|\tat'|$ separately.  
\begin{itemize}
\item The clusters in $\tao'$ are based on the \scnd instance ${\cal I}_1$ (in \mto), which  has a source for every cluster in \tao. In particular, each cluster in $\tao'$ corresponds to some subtree $X\in {\cal N}_1$ with root $r\not\in N=\{v_F:F\in \th\}$. See Step~\ref{step:t1-merge} in \mto.  As $r$ is not a neighbor of sink $t$, property 4 in Theorem~\ref{thm:ss-disj-cluster} implies that subtree $X$ contains at least two sources. It follows that $|\tao'|\le \frac12 |\tao|$.

\item The clusters in $\tat'$ are based on the \scnd instance ${\cal I}_2$ (in \mtt). Recall that we use a bi-partition $({\cal A}^+, {\cal A}^-)$ of $\tat$. Instance ${\cal I}_2$  has a source for each cluster in ${\cal A}^+$ and the neighbors $M$ of the sink correspond to clusters in ${\cal A}^-$. Let $B\sse M$ denote the neighbors of the sink that appear in some subtree of ${\cal N}_2$; let ${\cal B}\sse {\cal A}^-$ denote the corresponding clusters. The collection $\tat'$ consists of (i) clusters from ${\cal A}^-\setminus {\cal B}$, and (ii) clusters corresponding to subtrees in ${\cal N}_2$. By property 4 in Theorem~\ref{thm:ss-disj-cluster}, every subtree in ${\cal N}_2$  has at least two clusters from ${\cal A}^+ \cup {\cal B}$. So, the number of clusters in case (ii) above is at most $\frac12\cdot \left(|{\cal A}^+| + |{\cal B}|\right)$. Clearly, the number of clusters in case (i) is $|{\cal A}^-| - |{\cal B}|$. So,
$$|\tat'| \le \frac{|{\cal A}^+| + |{\cal B}|}{2} + |{\cal A}^-| - |{\cal B}| = \frac{|\tat|}{2} + \frac{|{\cal A}^-| - |{\cal B}|}{2}\le \frac34\cdot |\tat|,$$
where we used that $|{\cal A}^-|\le |\tat|/2$ as $|{\cal A}^-|\le |{\cal A}^+|$.
\end{itemize}
Thus, $|\tao'| +|\tat'|\le \frac12\cdot |\tao| + \frac34\cdot |\tat| \le \frac34\cdot (m_1+m_2)$.
\end{proof}


Finally, we show that 
the demand preserved in set $K$ at the end of algorithm \mcc is a constant fraction of the total demand in \ncnd. 
\begin{lemma}\label{lem:delete}
The total demand  at the end of  \mcc is $\sum_{i\in K} d_i \ge D/4$. 
\end{lemma}
\begin{proof}
Note that requests are only removed in Step~\ref{step:cluster-remove-req} of \mcc. 
We will account for the deleted requests by explicitly ``charging'' them to requests that will be preserved in $K$ (at the end of \mcc). 

Consider any cluster $T$ that is processed in  Step~\ref{step:cluster-remove-req} at any iteration of \mcc. Note that $T$ must be in the active set, i.e. $T\in \tao\cup\tat$. Let $t_i$ be the total demand of $T$'s internal terminals; note that the demand of $T$'s internal requests is $t_i/2$ as each such request has two terminals in $T$. Let  $t_h$ (resp. $t_a$) be the total demand of $T$'s external terminals with mates in \th (resp. $\tao\cup\tat$). By Lemma~\ref{lem:mc-cluster-structure}, there are no requests (in the current set $K$) from $T$ to any cluster of \ti. So, $\load(T)=t_i+t_h+t_a$. In this step, we remove (from $K$) all requests from $T$ to other active clusters. So, the deleted demand is exactly $t_a$. Furthermore, we claim that the following requests incident to $T$ will be preserved (in $K$) until the end.
\begin{itemize}
\item {\em Internal requests  in $T$.} This is because we never remove any internal request. Also, clusters only merge with each other during the algorithm: so any internal request in $T$ will remain internal to some cluster. 
 \item {\em External requests from $T$ to \th.} This is because we never remove any request incident to a heavy cluster. Also, any cluster that is currently heavy will remain heavy for the rest of the algorithm.
\end{itemize}
The total demand in these ``preserved'' requests is $\frac{t_i}2+t_h$. We now show that the removed demand $t_a\le 3\cdot (\frac{t_i}2+t_h)$. Consider the two cases when requests are removed:
\begin{itemize}
\item {\em $T$ is an internal cluster.} Then, we have $t_i\ge \frac{\load(T)}{2} = \frac{t_i+t_h+t_a}2$; so $t_a\le t_i$. 
\item {\em $T$ is a dangerous cluster.} Here, $T$ is t1-active, which means it has at least $\frac{\load(T)}{4}$ demand going to heavy clusters (recall that $T$ has no requests to internal clusters).  Then, we have $t_h\ge \frac{t_i+t_h+t_a}4$, which implies $t_a\le 3t_h$.
\end{itemize}
Thus, the total removed demand incident to $T$ is at most $3$ times the total preserved demand incident to $T$. 

Finally, we show that the preserved requests incident to different clusters $T$ that are processed in Step~\ref{step:cluster-remove-req} (\mcc) are {\em disjoint}.  To this end, we will show that once a cluster $T$ is processed in Step~\ref{step:cluster-remove-req}, it will not be part of any cluster $T'$ that is processed  in Step~\ref{step:cluster-remove-req} (and removes requests) at a later iteration. Indeed, after cluster $T$ is processed in Step~\ref{step:cluster-remove-req}, there are no requests from $T$ to other active clusters. If cluster $T$ becomes internal  then it remains unchanged  in later iterations. If cluster $T$ becomes t1-active then it may merge with other clusters: but these can only be heavy clusters (in which case $T$ also becomes part of that heavy cluster) or 
t1-active clusters (which by Lemma~\ref{lem:mc-cluster-structure} also have no external requests to active clusters). In either case, cluster $T$ will never be part of another cluster that gets processed in Step~\ref{step:cluster-remove-req} again. It now follows that the total removed demand (over all iterations) is at most $3$ times the total preserved demand, which completes the proof.
\end{proof}

\def\cf{\ensuremath{{\cal \widehat T}}\xspace}

\medskip \noindent {\bf Completing Proof of Theorem~\ref{thm:clustermc}.} We are now ready to prove the multicommodity clustering theorem.  Let ${\cal \widehat T}$ denote the final 
collection of  clusters  and $K\sse [k]$ the final set of requests, at the end of algorithm \mcc. 

{\em Property (i). Each cluster in \cf is internal or heavy.} By the termination condition, there is no active cluster remaining at the end, i.e., $\tao\cup\tat=\emptyset$. So, $\cf=\th\cup\ti$, which means all clusters are internal/heavy (see Lemma~\ref{lem:mc-clust-class}).

{\em Property (ii). Each terminal (i.e. node in $\{s_i,t_i\}_{i\in K}$) lies in exactly one cluster.} Initially, each terminal is in its own cluster: so this property is true. In algorithm \mcc, we  only merge clusters together by adding some subtrees (based on \scnd instances). Note that these subtrees do not contain any terminal because each terminal is a leaf-node, and the sources/sink in our \scnd instances are new nodes (different from the original terminals). So, each  terminal lies in exactly one cluster throughout  the algorithm.

{\em Property (iii). The total demand of request-pairs $K$ is $\sum_{i\in K} d_i\ge D/4$.} This follows directly from Lemma~\ref{lem:delete}.

{\em Property (iv). Each node appears in  at most $O(\log  k)$ different clusters.} We first claim that the number of iterations in \mcc is $O(\log k)$. Indeed, by Lemma~\ref{lem:mcc-iteration}(5), the number of active clusters drops by a  factor of $4/3$ in each iteration. As there are $2k$ active clusters initially, we will have no active clusters left after $O(\log k)$ iterations. 
By Lemma~\ref{lem:mcc-iteration}(1), in each iteration, each node gets added to at most two clusters. Combined with the number of iterations, property (iv)  follows. 

{\em Property (v). The demand of each cluster is at most $O(\congsingle^2 \log k)\cdot q$.} The demand of any internal/heavy cluster when it is formed is $O(\congsingle)\cdot q$ by Lemma~\ref{lem:mcc-iteration}(3). Note that internal clusters do not change after they are formed (and added to \ti). For any heavy cluster, by Lemma~\ref{lem:mcc-iteration}(4),  the increase in demand is $O(\congsingle^2)\cdot q$ in each iteration. As the number of iterations is $O(\log k)$, the final demand of any heavy cluster is $O(\congsingle^2 \log k)\cdot q$.

{\em Property (vi). The total cost  $\sum_{T\in {\cal \widehat T}}  \sum_{v\in T} c_v   \le   O(\costsingle \cdot \log k) \cdot \opt$.} The cost of nodes added in any iteration is $O(\costsingle)\cdot  \opt$ by Lemma~\ref{lem:mcc-iteration}(2). Combined with the $O(\log k)$ number of iterations, we obtain property (vi).  

This completes the proof of Theorem~\ref{thm:clustermc}.

\subsection{Routing Across Clusters \label{subsec:mc-routing}}
From Theorem~\ref{thm:clustermc}, we have a collection $\widehat{{\cal T}}$ of low-cost, nearly node-disjoint clusters  containing requests $K\sse [k]$. Here, we show how to route a constant fraction of the requests in $K$. Routing within a cluster can be done easily using the corresponding subtree. So we focus on routing each request  across clusters,  from their ``source cluster'' to their ``sink cluster''. In order to achieve this, we will add some nodes/edges to our solution. 

\def\cH{\ensuremath{{\cal H}}\xspace}
\def\cM{\ensuremath{{\cal M}}\xspace}
\def\cE{\ensuremath{{\cal E}}\xspace}

\paragraph{Algorithm Overview} Recall that $\widehat{{\cal T}}$ contains two types of clusters: internal and heavy. Let $\ti$ denote all internal clusters in $\widehat \T$, and $\th=\widehat{{\cal T}}\setminus \ti$ denote the heavy clusters. If a cluster is both internal and heavy then it is treated as an internal cluster.  
If $\ti$ contains a constant fraction of the demand in $K$, then we don't need to route across clusters: simply use the subtrees in $\ti$ to route all these internal requests. The hard   case is  when  most of the demand in $K$ is contained in $\th$. In this case, we use an idea from \cite{AntoniadisIKMNPS14} for the {\em edge} capacitated routing problem. Specifically, each request $i$ ``hallucinates'' that its demand equals $q$ with probability $\approx \log(k) \frac{d_i}{q}$ (and zero otherwise), and we find a subgraph \cH that supports all the hallucinated demands. We can find a good approximation to this ``hallucinated instance'' by rounding a natural linear-program (see \S\ref{mc-hallucination}). 
We then use the union of $\cH$ and $\widehat \T$ as our solution. However (unlike the edge-capacitated case), this solution may not suffice to route {\em all} the remaining demands. Nevertheless, we show that a constant fraction of the remaining demand can still be routed in $\cH\cup \widehat{\T}$. To this end,  
we partition the heavy clusters into ${\cal T}_1, {\cal T}_2, \cdots, \T_p$ such that the minimum (edge) cut of the demand graph induced on each ${\cal T}_{j}$ is at least $\Omega(q)$. This partitioning algorithm is based on iteratively removing {\em minimal} min-cuts (see \S\ref{mc-routable}). We also delete all requests in $K$ crossing from one part to another, and prove that the remaining demand is still a constant fraction of that in $K$.    Then, we show that when min-cuts are large, the hallucinated request-pairs behave like a \emph{cut-sparsifier}~\cite{karger1999random} of the demand graph (after contracting the clusters). So the hallucinated solution \cH has enough capacity to support an {\em edge-capacitated} routing across clusters. Finally, we ``un-contract'' these clusters by using the trees in $\th$ to route within each cluster (see \S\ref{mc-routing-halln}). We bound the {\em node congestion} of the routing using the fact that each cluster has load $O(\congsingle^2 \log k) \cdot q$.   The formal algorithm is given as Algorithm~\ref{alg:routing-mc}: it takes as input the clusters $\widehat{\T}$ (and requests $K$) from Theorem~\ref{thm:clustermc}, and outputs a subgraph $\cE$ of $G$ along with a subset $K'\sse K$ of requests that can be routed in \cE at low node congestion.

\begin{algorithm}
\caption {\ncnd Routing Algorithm \label{alg:routing-mc}}
\begin{algorithmic}[1]
\State \label{step:route-easy} if the internal requests of clusters in $\ti$ have demand at least $\frac16 \sum_{i\in K}d_i$,  return $\cE=\ti$ as the solution and all  internal requests in $\ti$ as the routable pairs $K'$.

 \State \label{step:route-mincut1} apply Theorem~\ref{thm:mc-mincut} to the heavy clusters $\th$ and obtain partition  ${\cal T}_1, {\cal T}_2, \cdots, \T_p$. 
\State \label{step:route-mincut2} let $K'\sse K$ denote the union of requests induced in $\T_j$ (for $j=1,2,\cdots p$).

\State  let $r=\frac{\alpha_1 \log k}q$ where $\alpha_1>1$ is some constant.
\State \label{step:route-sample} let $\cM$ denote the random instance with each request $i\in K$ having demand $B_i\cdot q$, where $B_i\sim \mathsf{Binomial}(d_i, r)$ independently. 
\State \label{step:route-halln} apply Theorem~\ref{thm:mc-unit} to obtain solution $\cH\sse V$ for the  hallucinated  instance  \cM.
\State \label{step:route-end} return $\cE=\cH\cup \th$ as the solution and $K'$ as the routable pairs.
\end{algorithmic}
\end{algorithm}

The rest of this subsection proves the following main result. 
\begin{theorem}\label{thm:mc-routing} Given any \ncnd instance with optimal cost \opt, after running Algorithms \ref{alg:cluster-mc} and \ref{alg:routing-mc}, we have the following with probability at least $1-O(\frac1{k^2})$:
\begin{itemize}
\item   The cost of solution \cE is $O(\costsingle \log k)\cdot \opt$.
\item Solution \cE  supports an unsplittable routing for all requests in $K'$ with node congestion $O(\congsingle^2\log^3k)$. 
\item The total demand of requests in $K'$ is at least $D/{24}$.
\end{itemize}
\end{theorem}

Throughout, we assume that $r=\frac{\alpha_1 \log k}q<1$: so $r$ is a valid probability value in step~\ref{step:route-sample}. If $r\ge 1$ then we have $q=O(\log k)$, in which case there is an easy $(O(\log n \log k),  O(\log n \log k))$ bicriteria approximation algorithm; see Appendix~\ref{app:small-q}.

\subsubsection{Identifying Routable Request-Pairs in Heavy Clusters\label{mc-routable}}
\def\K{\ensuremath{\widetilde{K}}\xspace}

We now show how to identify 
a subset $K'\sse K$ of request-pairs by partitioning the demand graph into components of high min-cut values. 
 Having a high min-cut in the demand graph  is necessary for the cut-sparsification  argument that is used (in \S\ref{mc-routing-halln})  to route requests  $K'$ with low congestion.

We first define a cluster multigraph as follows. 
\begin{definition} [Cluster Multigraph $\cC({\cal T})$] \label{def:cg}
Given a  collection ${\cal T}$ of clusters, the multigraph $\cC({\cal T})$ has a node for every cluster $T \in {\cal T}$, and for each request-pair $i\in K$,  an edge of weight $d_i$ between clusters $T_a, T_b\in {\cal T}$ where $s_i\in T_a$ and $t_i\in T_b$.
\end{definition}
 We now show how to partition $\cC(\T)$ into several {\em parts} so that the minimum (edge)  cut of each part is  $\Omega(q)$ and the demand of ``crossing'' requests is small.
 For any graph $J$ and subset $X$ of nodes, we use the  notation $\partial_J(X)$ to denote the set of edges in $J$ with exactly one end-point in $X$.
\begin{theorem}\label{thm:mc-mincut}
There is a polynomial-time algorithm that, 
given any  collection ${\cal T}$ of clusters,  computes a partition  ${\cal T}_1, {\cal T}_2, \cdots, \T_p$ of $\T$  such that:
\begin{enumerate}
\item[i.] For each $j\in[p]$, the induced cluster graph $\cC(\T_j)$ has min-cut at least $ q/8$.  
\item[ii.] The total weight of edges in $K'=\bigcup_{j=1}^p \cC(\T_j)$ is at least $W - \frac{N q}{4}$.  
\end{enumerate} 
Here, $N=|\T|$  and $W$ is the total weight of edges in $\cC(\T)$. 
\end{theorem}
\begin{proof}
Consider the following procedure to obtain the partition. Initially, $\T'=\T$  and  $K'$ consists of all edges in $\cC(\T)$. 
For $j=1,2,\cdots $ do:
 \begin{enumerate}
\item If the min-cut value in $\cC(\T')$ is at least $\frac{q}{4}$ then $\T_j\gets \T'$ and stop.
\item Let $S\sse \T'$ denote a {\em minimal} min-cut in graph $\cC(\T')$.
\item Set $\T_j\gets S$, $\T'\gets \T'\setminus S$ and $K'\gets K' \setminus \partial_{\cC(\T')}(S)$.
\end{enumerate}
Note that this procedure creates one part in each iteration. At any point, $\T'$ denotes the remaining set of clusters/nodes (that are still unassigned to parts) and $K'$ denotes the current set of non-crossing edges. (An edge is said to be crossing if its end points lie in different parts.) Let  $G'=\cC(\T')$ denote the current graph. 
Let $p$ denote the number of iterations, which is also the number of parts in the partition of $\T$.  At the end of this procedure, $K'$ equals the set of edges in $\bigcup_{j=1}^p \cC(\T_j)$, which are all the non-crossing edges.

We first prove condition (ii). 
Clearly, the number of iterations is at most $N$, the number of nodes in $\cC$. As $S$ is a min-cut, the total weight of the edges $\partial_{G'}(S)$ removed in any iteration, is at most $ q/4$. So  the total weight of all edges removed is at most $(  Nq)/4$.  So, the total weight of $K'$ (at the end) is at least $W-(  Nq)/4$.

We now prove condition (i).  Note that each part $\T_j$ is either (i) the set $S$ in some iteration above, or (ii) the final set $\T'$. Clearly, in the latter case, graph $\cC(\T_j)$ has min-cut at least $ q /4 \ge \frac{  q}{8}$. In the former case, consider the graph $G'=\cC(\T')$ and set $S\sse \T'$ in the iteration when $\T_j=S$ was created. If $|S|=1$ then there is nothing to prove as part $\T_j=S$ would have infinite min-cut value. Let $A\subset  S$ be any strict subset.  By minimality of $S$, the weights of $\partial_{G'}(A)$   and $\partial_{G'}(S\setminus A)$   are both at least $\frac{  q}{4}$.  
 Let $a$ (resp. $b$) denote the total weight of edges having one end-point in $A$ (resp. $S\setminus A$) and the other end-point in $\T'\setminus S$.  Also, let $x$ denote the total weight of edges having one end-point in $A$ and the other in $S\setminus A$. Note that the weight of $\partial_{G'}(A)$ is $x+a$, weight of $\partial_{G'}(S\setminus A)$ is $x+b$ and weight of $\partial_{G'}(S)$ is $a+b$. 
Combined with the observation above (by minimality of cut $S$),
$$x+a\ge \frac{  q}{4}, \quad x+b\ge \frac{  q}{4}, \quad a+b< \frac{  q}{4}.$$
It follows that $x \ge \frac{  q}{8}$, i.e. the weight of edges between $A$ and $S\setminus A$ is at least $\frac{  q}{8}$. As this holds for all strict subsets $A\subset S$, the min-cut of $G'[S]=\cC(\T_j)$ is at least $\frac{  q}{8}$.
\end{proof}

 In our algorithm, we apply this result to the collection of heavy clusters \th. 
Next, we show how to add some nodes $\cH$ to the heavy clusters (see \S\ref{mc-hallucination}) so that  all requests in $K'$ can be routed in the resulting solution $\cE=\th\cup \cH$ with low node congestion (see \S\ref{mc-routing-halln}).

\subsubsection{Hallucinating to Connect Heavy Clusters}
 \label{mc-hallucination}
Here, we show how to find a low-cost solution \cH for the  ``hallucinated instance'' \cM.  Recall that instance \cM has demand $B_i\cdot q$ for each request-pair $i\in K$, where $B_i\sim \mathsf{Binomial}(d_i, r)$ independently. Note that the expected demand of request $i$ is $\E[B_i]\cdot q=rd_iq=O(\log k)\cdot d_i$. We treat the $B_i\cdot q$ demand of each request $i$ as $B_i$ many copies (each with demand $q$); so these demands can be sent along $B_i$ different $s_i-t_i$ paths (each carrying $q$ units). Note that  all nodes in graph $G$ have capacity $q$. By scaling down all capacities/demands by $q$, we obtain an  equivalent instance with unit node-capacities and $B_i$ many (unit demand) requests between $s_i$ and $t_i$,  for all $i\in K$. For simplicity, we work with this scaled instance as \cM. 
We first show that (with high probability), there exists a solution to ${\cal M}$ of low cost and bounded node congestion, and then provide an algorithm to find such a solution.
\begin{lemma}\label{lem:mc-LP}
With probability at least $1-O(\frac{1}{n^2})$, there is an unsplittable routing $\{P^*_i\}_{i\in \cM}$ of the hallucinated requests where $\sum_{i\in \cM} \sum_{v\in P^*_i} c_v\le (\alpha_2 \log n)\cdot \opt$ and the node-congestion is at most $\alpha_2\cdot \log n$, where $\alpha_2=O(1)$.
\end{lemma}
\begin{proof}
Consider the optimal solution for the original \ncnd instance. Let $P^*_i$ denote the path used for sending $d_i$ units of flow between $s_i$ and $t_i$, for each request $i\in [k]$. We now consider the solution $X$ that sends $B_i$ units  of flow on the optimal path $P^*_i$ for each request $i\in K$. (Equivalently, each of the $B_i$ copies of request $i$ uses the same path $P^*_i$ to route its unit flow.) We  now show that this solution has $O(\log n)$ congestion with high probability.
 To see this, consider any node $v\in V$. By feasibility of the solution $\{P^*_i:i\in [k]\}$, we have $\sum_{i\in [k] : v \in P^*_i} d_i \le q$. The load on node $v$ in solution $X$ is $L_v:=\sum_{i\in K : v \in P^*_i} B_i$. As each $B_i$ is a Binomial random variable,  $L_v$ is the sum of independent $[0,1]$ random variables. The mean 
$$\E[L_v] = \sum_{i\in K : v \in P^*_i} \E[B_i] = \frac{\alpha_1  \log k}{q}\sum_{i\in K : v \in P^*_i} d_i \le \alpha_1 \log k\le \alpha_1 \log n.$$ By a Chernoff bound, there is a constant $\alpha_2$ such that $\Pr[L_v > \alpha_2 \cdot \log n] \leq \frac{1}{n^3}$. Taking a union bound over all $n$ nodes, we have $\Pr\left[\exists v \, : \, L_v > \alpha_2\cdot \log n\right] \leq \frac{1}{n^2}$. We condition on the event that $L_v\le \alpha_2 \cdot \log n$ for all $v\in V$. Then, the node-congestion of solution $X$ is as claimed. We now bound the cost: 
$$\sum_{i\in \cM} \sum_{v\in P^*_i} c_v = \sum_{v\in V} c_v \cdot L_v = \sum_{v\in \opt} c_v \cdot L_v \le (\alpha_2 \log n)\cdot \opt.$$
Hence solution $X$ satisfies both properties in the lemma. 
\end{proof}
 \begin{theorem}\label{thm:mc-unit}
There is a polynomial algorithm that finds a solution \cH for the hallucinated instance \cM satisfying the following with probability at least $1-O(\frac{1}{n^2})$:
\begin{itemize}
\item The total cost of nodes in $\cH$ is $O(\log n)\cdot \opt$.
\item The node congestion of \cH is $O(\log n)$.
\end{itemize}
\end{theorem}
\begin{proof}
 We use the following linear program relaxation:
\begin{alignat}{2}
&\min \sum_{i\in \cM} \sum_{p\in \cP_i} (\sum_{v\in p} c_v)\cdot f(p)  \tag{$LP_h$} \label{lp:cap} &\\
& \mbox{s.t.}  \; \sum_{p \in \cP_i} f(p) = 1  &\forall i\in {\cal M}  \label{f:constraintone}\\
&  \sum_{p | v \in p} f(p) \leq \alpha_2\cdot \log n\;\;\;\;\;  &\forall v \in V  \label{f:constrainttwo}\\
& f(p) \geq 0  &\forall i \in {\cal M}, \,\, \forall p \in \cP_i.
\end{alignat}
Here, $\cP_i$ is the set of all $s_i$-$t_i$ flow paths in $G$. Constraint~\eqref{f:constraintone} requires that exactly one path be selected for each request $i\in \cM$ (recall that \cM contains $B_i$ copies of each request $i\in K$). Constraint~\eqref{f:constrainttwo} bounds the node-congestion by $O(\log n)$. Note that the LP objective corresponds to the sum of costs over all paths in the solution (so each node gets counted multiple times). By Lemma~\ref{lem:mc-LP}, the optimal value of \eqref{lp:cap} is at most $(\alpha_2 \log n)\cdot  \opt$ with probability at least $1-O(\frac1{n^2})$. Although  this LP has an exponential number of variables, it can be re-formulated using a polynomial number of (flow-based) variables. Hence, we can solve this LP exactly in polynomial time. 

We now perform simple randomized rounding of the LP solution. For each request $i\in \cM$, select a random path $U_i\in \cP_i$  with probability $\{f(p)\}_{p\in \cP_i}$ independently. The solution $\cH = \cup_{i\in \cM} U_i$. We now prove that \cH satisfies the claimed properties with probability at least $\frac12$. 

For any path $p$, let $c(p)=\sum_{v\in p} c_v$ denote its total node cost. The cost of \cH is at most $C(\cH):= \sum_{i\in \cM} c(U_i)$; note that 
$\E[C(\cH)]$ equals the LP optimal value, which is at most $(\alpha_2 \log n)\cdot  \opt$. So, by Markov's inequality, with probability at least $\frac23$, the cost of \cH is at most $(3\alpha_2 \log n)\cdot  \opt$. 

We now bound the node congestion. For any node $v\in V$, let $L_v$ denote the number of paths in $\{U_i\}_{i\in\cM}$ containing $v$. Note that $L_v$ is the sum of independent $0/1$ random variables, with mean $\E[L_v] \le \alpha_2  \log n$ by \eqref{f:constrainttwo}. By a Chernoff bound, there is a constant $\alpha_3$ such that $\Pr[L_v > \alpha_3 \cdot \log n] \leq \frac{1}{n^3}$. Taking a union bound over all nodes $v$, we have $\Pr\left[\exists v \, : \, L_v > \alpha_3\cdot \log n\right] \leq \frac{1}{n^2}$. 

Hence, with probability at least $\frac12$, we obtain both the claimed properties of \cH. We can boost the success probability by repeating this algorithm independently $O(\log n)$ times and returning the best solution found as \cH. This proves that \cH satisfies the claimed properties with probability at least $1-O(\frac{1}{n^2})$.
\end{proof}

\def\G{\ensuremath{\widetilde{\cC}(j)}\xspace}
\def\H{\ensuremath{\widetilde{H_C}(j)}\xspace}

\subsubsection{Routing flow in hallucinated graph} \label{mc-routing-halln}
Here, we show that all requests in $K'$ can be routed in our solution $\cE=\th\cup \cH$ with low congestion. Recall that $K'$ is the set of ``routable requests'' identified in \S\ref{mc-routable} and $\cH$ is the hallucinated solution found in \S\ref{mc-hallucination}. 
Also, recall the partition ${\cal T}_1,\cdots ,{\cal T}_p$ of   heavy clusters obtained by applying Theorem~\ref{thm:mc-mincut}. 
 For each part ${\cal T}_j$,   define a
random edge-capacitated graph  $H_C(j)$ as follows. Nodes of $H_C(j)$ correspond to clusters of ${\cal T}_j$. For each request  $i$ with both $s_i$ and $t_i$ in clusters of $\T_j$,  there are $B_i$ \emph{edges of capacity $q$} in $H_C(j)$ between the clusters containing $s_i$ and $t_i$;  these edges correspond to the $B_i$ many $s_i-t_i$ paths found in the hallucinated instance \cM (see Theorem~\ref{thm:mc-unit}). 

Henceforth, we shall slightly abuse notation and refer to the cluster graph $\cC({\cal T}_{j})$ as just $\cC(j)$. Recall that each edge in $\cC(j)$ corresponds to some request $i\in K$ with both $s_i$ and $t_i$ in $\T_j$, and has weight $d_i$ (the demand of request $i$).  

We  will make use of the following cut-sparsification result.
\begin{theorem}[Theorem~2.1 \cite{karger1999random}] \label{thm:karger} 
Let $G$ be an $N$-node  multigraph with min-cut $\kappa$, and $r\in[0,1]$. Let $H$ be a multigraph containing each edge of $G$ independently with probability $r$. If $r\cdot \kappa\ge \frac{3(d+2)\ln N}{\epsilon^2}$ for some $d,\epsilon$, then with probability $1-O(1/N^{d})$, every cut in $H$ has value within $r(1\pm \epsilon)$ of the cut value in $G$.
\end{theorem} 

 \begin{lemma}\label{lem:Karger}
For any $j\in [p]$, with probability at least $1-O(\frac{1}{k^3})$, all request-pairs in $\cC(j)$ can be routed fractionally in $H_C(j)$ without exceeding edge capacities.
\end{lemma}
\begin{proof}
By Theorem~\ref{thm:mc-mincut}(i), the minimum cut in $\cC(j)$ has value $\kappa\ge \frac{  q}{8}$. Let $\G$ be an {\em unweighted} multigraph obtained from $\cC(j)$ by replacing each edge $i$ in $\cC(j)$ with 
$d_i$ parallel edges ($d_i$ is the demand of request $i$). Note that for any subset $S\sse \T_j$, its cut values in $\cC(j)$ and \G are the same. 
So the 
min-cut in \G is also $\kappa\ge \frac{  q}{8}$. Note that $H_C(j)$ can be viewed equivalently as a random subgraph of \G obtained by selecting each edge independently with probability $r=\frac{\alpha_1 \ln k}{q}$, and assigning capacity $q$ to each selected edge.

 We now apply Theorem~\ref{thm:karger} on graph \G with $r$ and $\kappa$ as above. Note, $r\cdot \kappa\ge \frac{\alpha_1  }{8}\ln k\ge 60\ln k$, 
assuming that $\alpha_1\ge  480$. So, we can set $d=3$ and $\epsilon=\frac12$ in this theorem, which implies that with probability at least $1-O(\frac{1}{k^3})$, 
for every subset $S\sse \T_j$, its cut value in $H_C(j)$ is at least $\frac{\alpha_1 \ln k}{2}$ times its cut value in $\cC(j)$.

Note that each edge in $H_C(j)$ has capacity $q$: so the cut value of $S$ in $H_C(j)$ is $cut_H(S):=q\cdot |\partial_{H_C(j)}(S)|$. Let $cut_G(S):=  \sum_{i\in \partial_{\cC(j)}(S)} d_i$ denote the cut value of $S$ in $\cC(j)$. 
In other words, the {\em non-uniform sparsest cut} of the  multicommodity routing instance with demand-graph $\cC(j)$ and capacity-graph $H_C(j)$ is:
$$\min_{ S\sse \T_j} \frac{\mbox{capacity across }S}{\mbox{demand across }S}  =  \min_{ S\sse \T_j } \frac{cut_H(S)}{cut_G(S)} \ge  \frac{\alpha_1 \ln k}{2}.$$
 Choosing $\alpha_1$ to be a large enough constant, we can ensure that the above sparsest cut value is  at least the multicommodity flow-cut gap $\Theta(\log k)$~\cite{Linial-London-Rabinovich}. This proves the existence of a fractional routing for request-pairs in $\cC(j)$.
\end{proof}
This lemma enables us  to find an {\em edge-capacitated} multicommodity flow for the request pairs $K'$ in $\bigcup_{j=1}^p \cC(j)$. For each request $i$  in $\cC(j)$, let $f_i$ denote the fractional flow sending $d_i$ units from the source to sink cluster  in graph $H_C(j)$. Note that the flows $\{f_i\}_{i\in K'}$ can be routed concurrently, while respecting all edge capacities: so the total flow through each edge of $\bigcup_{j=1}^p H_C(j)$ is at most $q$. Further, the next lemma shows that the total flow  through any node  is also bounded.

\begin{lemma}\label{cl:mc-node-ub}
The total capacity of edges in $\bigcup_{j=1}^p H_C(j)$  incident to any node is $O(\congsingle^2 \log^2 k)\cdot q$, with  probability at least $1-O(\frac{1}{k^3})$.
\end{lemma}
\begin{proof} 
 Consider any node (cluster) $T\in H_C(j)$ for any $j\in[p]$. From Theorem~\ref{thm:clustermc}(iv), we know that  the total demand of requests $R_T$ incident to $T$ (i.e. having a terminal in $T$) is $O(\congsingle^2 \log k)q$. Every edge in $H_C(j)$ incident to $T$ corresponds 
 some  request $i\in R_T$, and the edge has capacity  $B_i\cdot q$, where $B_i\sim \mathsf{Bernoulli}(d_i,r)$. So the total capacity incident to node $T$ is $X:=\sum_{i\in R_T} B_i\cdot q$. Note that  $\E[X] =qr\sum_{i\in R_T} d_i=\alpha_1\log k\sum_{i\in R_T} d_i = O(\congsingle^2 \log^2 k)\cdot q$. As $X$ is the sum of independent $\{0,q\}$ random variables, we obtain by a Chernoff bound, that 
$X= O(\congsingle^2\log^2 k)\cdot q$ with  probability at least $1-\frac{1}{k^4}$. Finally, a  union bound over all nodes/clusters completes the proof.
\end{proof}

\smallskip
\noindent {\bf Obtaining node-capacitated routing in $G$.}  Now, we show that the edge capacitated  routing $\F=\{f_i:i\in K'\}$ in the hallucinated graph $H_C:=\bigcup_{j=1}^p H_C(j)$ can also be implemented as a node-capacitated routing in the real graph $G$. This involves un-contracting nodes of $H_C$ into clusters $\th$ and the edges of $H_C$ into flow-paths of the hallucinated solution \cH. See Figure~\ref{fig:overall-route} for an example.

\begin{figure}[h]
\centering
\includegraphics[scale=0.68 ]{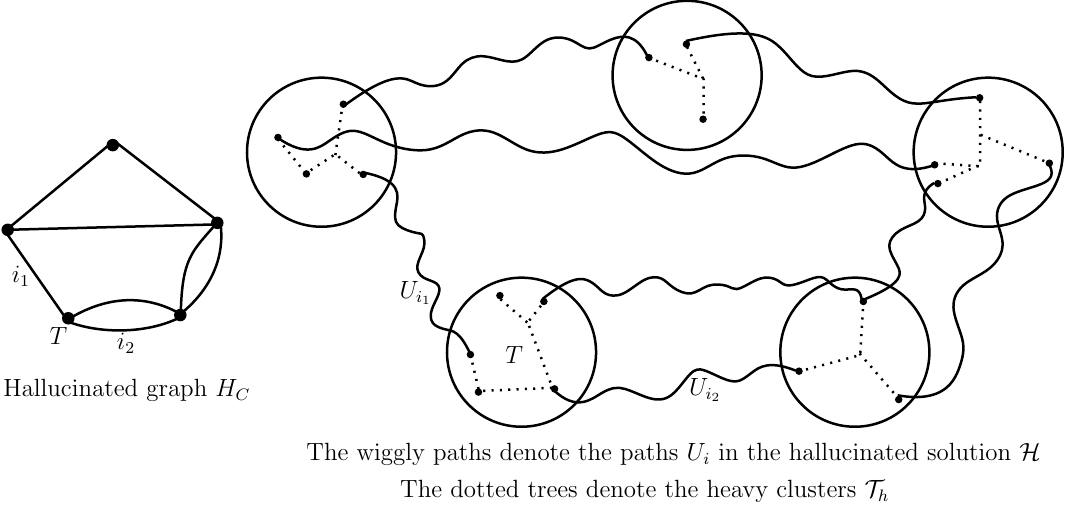}
\caption{ Un-contracting hallucinated graph $H_C$ using clusters $\th$ and flow \cH.\label{fig:overall-route}}
\end{figure}

\begin{lemma}\label{lem:mc-route-cong} With probability at least $1-O(\frac1{k^2})$,   solution $\cE= \cH\cup \th$ in step~\ref{step:route-end} of Algorithm~\ref{alg:routing-mc} supports an unsplittable routing of  requests $K'$ with node-congestion $O(\congsingle^2\log^3 k)\cdot q$. 
\end{lemma}
\begin{proof}  
We start with the fractional routing $\F$ in the hallucinated graph $H_C$, which corresponds to routing flow across clusters. By Lemma~\ref{lem:Karger}, the total flow on each edge of $H_C$ is at most $q$, with probability at least $1-O(\frac{1}{k^3})$. We replace the flow in $\F$ on each 
edge $i$ of $H_C$ 
with the flow-path $U_i$ used for request $i$ in the hallucinated solution \cH (recall that $i$ corresponds to some request in \cM). By Theorem~\ref{thm:mc-unit}, each node in graph $G$ appears in $O(\log n)$ many flow-paths $\{U_i\}$. 
 So, this results in a flow of $O(\log n)\cdot q$ through each node of graph $G$. However, we do not yet have a valid routing as we still need to route flow within each cluster (i.e., node of $H_C$).  

In order to route flow in $\F$ through any node/cluster $T\in H_C$, we do the following.  Consider any pair of consecutive edges $i_1,i_2$ (both incident to cluster $T$) used in the routing $\F$. Although 
the end-points of paths $U_{i_1}$ and $U_{i_2}$ may be different, they must both be terminals in cluster $T$: so we can use the subtree corresponding to $T$ to route the flow through this cluster. See Figure~\ref{fig:overall-route}. 
 By Lemma~\ref{cl:mc-node-ub}, the total flow in $\F$ through any node of $H_C$ is $O(\log^2 k) \congsingle^2 \cdot q$, with probability at least $1-O(\frac{1}{k^3})$. 
Moreover, by Theorem~\ref{thm:clustermc}(iii), each node in graph $G$ appears in $O(\log k)$ many clusters. Hence, the flow within clusters can be implemented so that the flow through each node is $O(\log^3 k) \congsingle^2 \cdot q$.

Combining the flow routing across clusters and within  each cluster, we obtain a valid fractional flow in graph $G$ with node congestion
$O(\log n)\cdot q+O(\log^3 k) \congsingle^2 \cdot q=  O(\log^3 k) \congsingle^2 \cdot q$; here we used the fact that $\gamma\ge \log n$.

Finally, we perform simple randomized rounding to obtain an unsplittable routing. Using the fact that each demand is at most $q$ and Chernoff bound, the node congestion in $G$ remains  $O(\congsingle^2\log^3 k) q$ with probability at least $1-\frac{1}{n^3}$. 
\end{proof}

\subsubsection{Completing  Proof of Multicommodity Routing} We now combine the results from \S\ref{mc-routable}, \S\ref{mc-hallucination} and \S\ref{mc-routing-halln} to complete the proof of the routing algorithm  (Theorem~\ref{thm:mc-routing}). 
Let $D'=\sum_{i\in K} d_i$ denote the total demand of the requests in $K$ that are obtained after Algorithm~\ref{alg:cluster-mc}. By Theorem~\ref{thm:clustermc}(i), we have $D'\ge D/4$ where $D$ is the total demand of all requests in the \ncnd instance. Let $K_1\sse K$ denote all requests $i\in K$ that have both $s_i$ and $t_i$ in the {\em same} internal cluster $T\in \ti$. Let  $K_2\sse K$    denote all 
requests $i\in K$ with  both $s_i$ and $t_i$ in some heavy cluster (the source/sink can be in different clusters of $\th$). Let $K_0=K\setminus K_1\setminus K_2$ be all other requests. We use $D_1$, $D_2$  and $D_0$ to denote the total demand of the respective sets. Let $N=|\th|$ be the number of heavy clusters.  Then, we have:

\begin{claim}\label{cl:cases}
If $D_1< \frac{D'}{6}$  then $\left(D_2-\frac{N q}{4}\right) \ge \frac{D'}{6}$.
\end{claim}
\begin{proof}
 For  every request  $i\in K_0$, either  $s_i$ or $t_i$  appears as an external terminal in some cluster $T\in \ti$. This implies that $D_0$ is at most the total demand of external terminals of $\ti$. By definition of internal clusters (Definition~\ref{def:mccluster}), the external demand of  any $T\in \ti$ is less than $\load(T)/2$ and the internal demand of $T$ is at least $\load(T)/2$. Adding over all internal clusters $T\in \ti$, we obtain  $D_0\le \frac12 \sum_{T\in \ti} \load(T)$ and 
$D_1 \ge \frac{1}{2}\sum_{T\in \ti}\frac{\load(T)}2$ (the factor $2$ reduction is because each internal request contributes twice to the cluster's load). Hence, $D_0\le 2D_1$. 

By definition of heavy clusters, the total demand in each $T\in \th$ is at least $q$.  Adding over all $T\in \th$, we get $D_0+2D_2\ge q N$, where we used that the sum of demands over $\th$ is at most $2D_2+D_0$.
  
Now, using $D'=D_0+D_1+D_2$ and $D_0\le 2D_1$, we obtain $D_2\ge D'-3D_1 > \frac12 D'$ as $D_1< \frac{D'}{6}$. Also,  $D_0\le 2D_1 <\frac13 D' <\frac23 D_2$. Now we have $ N q \le D_0 + 2D_2 <\frac83 D_2$. So, $D_2-\frac{ N q}{4} > D_2-\frac23 D_2 > \frac{ D'}6$, which proves the claim. 
\end{proof}

{\bf Large internal demand.}  We first consider the easier case that $D_1\ge D'/6$. So, 
step~\ref{step:route-easy} applies in Algorithm~\ref{alg:routing-mc}. Here, our solution $\cE=\ti$ and requests $K'=K_1$. Note that each cluster $T\in \ti$ can support all its internal demands on the tree corresponding to $T$ with node-congestion $\load(T)$. The cost of \cE is $O(\costsingle \log k)\cdot \opt$ by Theorem~\ref{thm:clustermc}(v). Moreover, each node of $G$ appears in $O(\log k)$ many clusters  (Theorem~\ref{thm:clustermc}(iii)) and $\load(T)=O(\congsingle^2 \log k)q$ for each cluster $T$ (Theorem~\ref{thm:clustermc}(iv)). So the node-congestion of this routing is $O(\congsingle^2 \log^2 k)q$. 

{\bf Large external demand.}  We now consider the  case that $D_1<  D'/6$. Here, our solution $\cE=\cH\cup \th$ where \cH is the solution to the hallucinated instance in step~\ref{step:route-halln}. The requests $K'$ are those obtained in Theorem~\ref{thm:mc-mincut}, which implies  $K'\sse K_2$ and its total demand is at least $D_2 - \frac{N q}{4}\ge D'/6$ (the last inequality is by Claim~\ref{cl:cases}). By Theorem~\ref{thm:mc-unit}, the cost of \cH is $O(\log n)\cdot \opt$ with probability   $1-\frac{1}{n^2}$. And, the cost of $\th$ is $O(\costsingle \log k)\cdot \opt$ by Theorem~\ref{thm:clustermc}(v). So the total cost of \cE is $O(\costsingle \log k)\cdot \opt$, where we use $\costsingle \ge \log n$. By Lemma~\ref{lem:mc-route-cong}, with probability  $1-\frac{1}{k^2}$, requests $K'$ can be routed  with  node-congestion $O(\congsingle^2 \log^3 k)q$. 

Thus, in either case, we have with probability at least $1-\frac{1}k$:
\begin{itemize}
\item  the cost of solution \cE is $O(\costsingle \log k)\cdot \opt$. 
\item requests $K'$ can be routed in \cE with node-congestion $O(\congsingle^2 \log^3 k)q$. 
\item the total demand in $K'$ is at least $D'/6\ge D/24$.
\end{itemize}

\subsection{Wrapping Up}
Our algorithm for \ncnd  invokes Algorithms \ref{alg:cluster-mc} and \ref{alg:routing-mc} iteratively $O(\log k)$ many times. By 
Theorem~\ref{thm:mc-routing}, each iteration results in  cost $O(\costsingle \log k)\cdot \opt$, node-congestion $O(\congsingle^2 \log^3 k)q$, and routes a constant fraction of the total remaining demand. Hence, after $O(\log D)$ iterations, we would have routed all the demands. As described in Appendix~\ref{app:poly-q}, we can ensure that $D$ is polynomial in $k\le n$. So, the number of iterations is $O(\log k)$. The final cost is $O(\costsingle \log^2 k)\cdot \opt$ and node-congestion is $O(\congsingle^2 \log^4 k)q$. By Theorem~\ref{thm:single-sink}, we have $\costsingle = O(\log^2n)$ and  $\congsingle=O(\log^3 n)$, which  implies that our cost is $O(\log^2n \log^2k)\cdot \opt$ and node-congestion is $O(\log^6n \log^4k)\cdot q$. This completes the proof of 
Theorem~\ref{thm:multicomm}.

\section{Conclusions}
In this paper, we obtained the first poly-logarithmic bicriteria approximation algorithm  for the uniform node-capacitated network design problem. There is still a large gap between our approximation bounds and the $\Omega(\log k)$ hardness of approximation that follows from node-weighted Steiner tree. Closing this gap is an interesting open question. Another question concerns non-uniform capacities, which is also open in the edge-capacitated case.  
A key challenge in dealing with non-uniform capacities is that there is no clear notion of how much demand should be aggregated in one cluster. 

\bibliographystyle{plain}
\bibliography{hallucination}

\appendix
\section{Simple Reductions for \ncnd and \eevrp}\label{app:reduction}

\subsection{Reducing \eevrp to \ncnd }\label{app:eevrp}
Here, we show that the energy-efficient routing problem can be reduced to the capacitated network design problem. This essentially follows from \cite{AndrewsAZ10}, but we provide the details for completeness. 
 
Consider any instance of \eevrp with energy function~\eqref{eq:energy}, graph $G=(V,E)$ having node-costs $\{c_v\}_{v\in V}$ and requests $\{(s_i,t_i,d_i)\}_{i=1}^k$. Let $q:=\sigma^{1/\alpha}$. Let $D_1=\{i\in  [k] : d_i\le q\}$ denote all requests with demand at most $q$, and $D_2=[k]\setminus D_1$ denote the ``large'' demands. We will handle these two demand types separately. Let \opt denote the optimal value. 

For small demands ($D_1$) we define an \ncnd instance on graph $(V',E')$ where $V'$ contains $k$ copies of each node in $V$. For each $(u,v)\in E$, there is an edge in $E'$ between every copy of $u$ and every copy of $v$.  
For any node $v\in V$ and $i\in [k]$, the $i^{th}$ copy of $v$ in $V'$ has cost $c_v\sigma( i^\alpha- (i-1)^\alpha)$. The requests remain the same (we can use any copy of the source/sink nodes). We claim that the optimal value of the \ncnd instance is at most $2^\alpha\cdot \opt$.  
To see this, consider the same routing as for the optimal \eevrp solution. If the flow through any node $v$ is $x$ then we include the {\em first} $\lceil x/q\rceil$ copies of $v$ into the \ncnd solution. Note that the cost of these copies of node $v$ is:
$$c_v \sigma  \lceil x/q\rceil^\alpha \le c_v \sigma \left( 1 + x/q \right)^\alpha = c_v (q+x)^\alpha \le c_v 2^\alpha (q^\alpha +x^\alpha) = 2^\alpha \cdot c_v\cdot f(x),$$
which is $2^\alpha$ times the cost of node $v$ in the \eevrp solution. 
Adding over all nodes, the total cost of this \ncnd solution is at most $2^\alpha\cdot \opt$.

Let $U\sse V'$ denote a $(\beta,\gamma)$ bicriteria approximate solution for \ncnd. Then, there is a feasible multicommodity flow that uses capacity at most $\gamma\cdot q$ on each node of $U$. By solving the natural LP and random rounding, we can also find (in polynomial time) a flow with load at most $L:=O(\gamma+\log n)\cdot q$ on each node of $U$. We now bound the \eevrp cost of this flow. Consider any node $v\in V$ that has $i$ copies selected in $U$: so the total cost of these nodes is $c_v \sigma i^\alpha$. The energy cost on $v$ is at most 
$$c_v(\sigma +    (iL)^\alpha) = O(\gamma^\alpha + \log^\alpha n)\sigma i^\alpha\cdot c_v.$$
Adding over all $v\in V$, the total energy cost is $O(\gamma^\alpha + \log^\alpha n)$ times the \ncnd cost, which is $O(\gamma^\alpha\cdot \beta)\cdot \opt$. Here, we assumed that $\gamma\ge \log n$.

For large demands ($D_2$) we just find an unsplittable routing that 
 minimizes the $\alpha^{th}$ power of loads. Let $\I'$ denote this problem instance. Note that this problem  differs from \eevrp only in the definition of the energy function, which is now $f'(x)=x^\alpha$  instead of $f(x)=\sigma+x^\alpha$. There is an $\alpha^\alpha$-approximation algorithm for this problem~\cite{MakarychevS18}.  Clearly, the optimal value of $\I'$  is at most \opt. Let $\tau$ denote an approximate solution to $\I'$ and $U\sse V$ denote the nodes carrying positive flow.  
As every request in $D_2$ has demand at least $q$ (and we have unsplittable flows),  every node in $U$ has flow {\em at least }$q$. Using the fact that $f(x)\le 2f'(x)$ for all $x\ge q$, it now follows that the \eevrp cost of $\tau$ is at most {\em twice} the $f'$-cost of $\tau$, i.e. at most $2\alpha^\alpha\cdot \opt$.

Combining the routing for $D_1$ and $D_2$ completes the proof of Theorem~\ref{thm:reduction}. 

\subsection{Approximation ratio relative to splittable routing}\label{app:splittable}
In our node-capacitated network design problem (\ncnd and \scnd), our goal is to find  an {\em unsplittable} routing of each demand. We now observe that our approximation guarantees in Theorems \ref{thm:single-sink} and \ref{thm:multicomm} are stronger, and hold relative to an optimal solution that only supports a {\em splittable} (i.e., fractional) flow of the demands.  

For \scnd, we only use the optimal solution \opt in Lemma~\ref{lem:cluster}. Observe that this relies on applying the confluent flow result (Theorem~\ref{thm:conflu}) to the optimal solution, which only requires a splittable flow.  

For \ncnd, we use the optimal solution \opt in the following steps:
\begin{itemize}
\item Bounding the optimal cost of the \scnd instances ${\cal I}_1$ and ${\cal I}_2$ (Lemmas~\ref{lem:lowcostmc1} and \ref{lem:lowcostmc2}). Here, we only use \opt to demonstrate a feasible fractional routing for ${\cal I}_1$ and ${\cal I}_2$: so we can also use a splittable-routing solution. 
\item Bounding the optimal cost of the hallucinated instance (Lemma~\ref{lem:mc-LP}). Here, we used the optimal paths $P^*_i$  to route a random quantity $B_i$ between $s_i$ and $t_i$, for each request $i\in [k]$. Then, we used a Chernoff bound to prove that the node congestion is $O(\log n)$ with high probability. Instead of an integral path $P^*_i$, we can use a fractional flow as a distribution ${\cal F}_i$ over $s_i-t_i$ paths for each $i\in [k]$. Then, we can first sample a random path $P_i$ from ${\cal F}_i$ and route the $B_i$ units on this path $P_i$. Again, a Chernoff bound can be used to prove that the node congestion is $O(\log n)$ with high probability. So, this step can also be carried out relative to an optimal splittable routing.
\end{itemize} 

\subsection{Reducing edge-costs to node-costs}\label{app:edge}
Here, we observe that the node version \eevrp is more general than the edge-version studied previously~\cite{AndrewsAZ10,Medalg,AntoniadisIKMNPS14}. Consider an instance of energy-efficient routing with edge energy costs~\eqref{eq:energy}, graph $G=(V,E)$ having edge-costs $\{c_e\}_{e\in E}$ and requests $\{(s_i,t_i,d_i)\}_{i=1}^k$. We define an \eevrp instance on the graph $G'$ obtained by subdividing each edge $e\in E$ with a node $v_e$. The node costs are zero for all nodes of $V$ and $c_e$ for each node $v_e$ (for $e\in E$). This \eevrp instance is clearly equivalent to the original edge-version.

A similar reduction shows that edge-capacitated \ncnd is a special case of the node-capacitated problem studied in this paper.

\section{Missing Details from Section~\ref{sec:mc}}

\subsection{Approximation Algorithm for Small $q$}\label{app:small-q}
Here, we provide a bicriteria approximation algorithm for \ncnd when $q$ is small. The idea is essentially the same as that used in Theorem~\ref{thm:mc-unit}. Given an \ncnd instance $\I$, consider a new instance $\I'$ where every demand equals $1$ and the goal is to select $s_i-t_i$ paths $P_i$ with node-congestion at most $q$ that minimizes the {\em sum} of all path costs, i.e. $\sum_{i=1}^k \sum_{v\in P_i} c_v$. Note that each node may be counted multiple times in this objective.  Using an optimal solution to $\I$ as a feasible solution to $\I'$, we have  $\opt(\I')\le q\cdot \opt(\I)$. We now write an LP relaxation for $\I'$ that is just \eqref{lp:cap}  used in Theorem~\ref{thm:mc-unit}, where $\cM=[k]$ and the right-hand-side in constraint~\eqref{f:constrainttwo} is $q$. It is clear that this is a valid relaxation. The rounding algorithm is the same as described in Theorem~\ref{thm:mc-unit}. The same analysis implies that we obtain a solution to $\I'$ of cost $O(\log n)\cdot \opt(\I')\le O(q\log n)\cdot \opt(\I)$ and node-congestion $O(q \log n)$. Using this as a solution to $\I$, the cost remains $O(q\log n)\cdot \opt(\I)$ and the node-congestion increases by at most a factor $q$. So we obtain an $(O(q\log n), O(q\log n))$ bicriteria approximation algorithm for \ncnd. 

We remark that the case when every demand is large, i.e., $\min_i d_i = \Omega(q)$, can also be solved by this approach. We just uniformly scale all demands and capacity by $\min_i d_i$ so that the new capacity $q'=O(1)$. Then, we obtain $(O(\log n), O(\log n))$ bicriteria approximation algorithm for \ncnd with large demands. 
\subsection{Ensuring Polynomially Bounded Demands}\label{app:poly-q}
Here, we show that the total demand $D=\sum_{i=1}^k d_i$ can be ensured to be polynomial in the number of requests $k$.    Let $D_1\sse [k]$ denote all requests with demand at least $q/k$, and $D_2=[k]\setminus D_1$. We will handle the requests in $D_1$ and $D_2$ separately. We round-up the demand $d_i$ of each $i\in D_1$ to an integer multiple of $q/k$, which increases each demand by at most a factor two. Then, scaling all demands down by $q/k$, we obtain an equivalent \ncnd instance with capacity $q' = O(k)$, which implies total demand in this instance is $D'=O(k^2)$.  
For requests in $D_2$, we just use the minimum node-weighted Steiner forest, which admits an $O(\log k)$-approximation algorithm~\cite{KleinR94}.

\section{Algorithm for LLSC}\label{app:llsc}
\def\C{{\cal C}}

We now outline the proof of Theorem~\ref{thm:llsc}, which  is based on the algorithm in Section 5 of \cite{BabenkoGGN12}. Recall that instance of LLSC consists of a set system $(U, {\cal C})$ with element-costs $\{c_v :v\in U\}$, bound $p\ge 1$, and two subsets: {\em required} elements  $W\sse U$ and {\em capacitated} elements $L\sse U$. The goal is to find a set cover ${\cal C'} \subseteq {\cal C}$ for the required elements $W$ (i.e. $\cup_{S\in {\cal C'}} S \supseteq W$) such that each capacitated element $e\in L$ appears in at most $p$ sets of ${\cal C'}$ and the cost $c({\cal C}') =\sum_{S\in {\cal C}'} \sum_{v\in S}c_v$  is minimized.

Let $\epsilon=1/(24\rho)$ where $\rho$ is the approximation ratio for the min-ratio oracle. We assume that we are given  an upper bound $T$ on the cost of the solution, and the goal is to find (if possible) a solution of cost at most $T$. (In order to find the min-cost solution, we can then perform binary search on $T$.)   The algorithm proceeds in several rounds $t=1,2,\cdots$, where one set from $\C$ is added to the solution in each round. In each round $t$, we define the following quantities:
\begin{itemize}
\item let $X\sse W$ be the  set of required elements covered in previous rounds.
\item for each capacitated element $e\in L$, let $A_t(e)$ be the number of times $e$ has been covered in previous rounds.
\item the potential in round $t$ is $\Phi(t)=\sum_{e\in L} (1+\epsilon)^{A_t(e)/p}$. 
\item the $f$-cost of any capacitated element $v\in L$ is $f_v = (1+\epsilon)^{A_t(v)/p}\left((1+\epsilon)^{1/p}-1\right) $; the $f$-cost of any other element $v\in U\setminus L$ is $f_v=0$. These costs correspond to bounding the $\ell_\infty$ norm of the loads on capacitated elements.
\item the $d$-cost of any element $v\in U$ is $d_v= \frac{\epsilon  \Phi(t)}{T}\cdot c_v$. These costs correspond to minimizing the $\ell_1$-norm of the original cost. The scaling factor of $\frac{\epsilon  \Phi(t)}{T}$ is used to normalize the contributions from the $\ell_\infty$-norm (of loads) and $\ell_1$-norm (of original costs).
\item the combined cost of any element $v\in U$ is $\eta_v = f_v +d_v$.  
\end{itemize}
We then select a set $S\in \C$ that approximately minimizes the ratio of its combined cost to the number of newly covered elements:
$$\frac{\sum_{v\in S} \eta_v}{|S\cap(W\setminus X)|}.$$
This is exactly the min-ratio problem, for which we have a $\rho$-approximation algorithm.

The analysis in \cite{BabenkoGGN12} shows that if there exists a solution of cost at most $T$ and load at most $p$ (on capacitated elements) then  the above algorithm has cost at most  $O(\rho \, \log|U|)\cdot T$ and load at most $O(\rho \, \log|U|)\cdot p$. Combined with a binary-search over the cost bound $T$, this completes the proof of Theorem~\ref{thm:llsc}.
\end{document}